\documentclass[%
endfloats*,
preprint,
nofootinbib,
amsmath,amssymb,
aps,
floatfix,
]{revtex4-1}

\usepackage{graphicx} 
\graphicspath{ {Figures/} }
\usepackage{epstopdf}
\usepackage[caption=false]{subfig}
\usepackage{amsmath}
\usepackage{amsfonts}
\usepackage{amssymb}
\usepackage{bm}
\usepackage{hyperref}
\usepackage{amsthm}
\usepackage{color}
\usepackage{mathrsfs}
\usepackage[vlined,ruled]{algorithm2e}
\usepackage{url}
\usepackage{color}
\usepackage{algorithmic}
\usepackage{bbm}
\usepackage{booktabs}
\usepackage{array}
\usepackage[table]{xcolor}
\usepackage{yfonts}
\usepackage{lineno}
\usepackage{mathdots}

\newcommand{\subscr}[2]{{#1}_{\textup{#2}}}

\newcommand{\real}{\mathbb{R}}

\newcommand{\transpose}{\mathsf{T}} 

\newcommand{\mc}{\mathcal}

\newcommand{\methods}{Methods}
\newcommand{\SI}{Supplement}

\DeclareSymbolFont{bbold}{U}{bbold}{m}{n}
\DeclareSymbolFontAlphabet{\mathbbold}{bbold}

\begin{document}
\bibliographystyle{naturemag}
\title{Data-Driven Control of Complex Networks}
\author{Giacomo Baggio}
\affiliation{Department of Information Engineering, University of Padova, Padova, Italy}
\author{Danielle S.~Bassett}
\affiliation{Departments of Bioengineering, Physics \& Astronomy, Electrical
  \& Systems Engineering, 
  Neurology, and Psychiatry, University of Pennsylvania, Philadelphia, USA\\ Santa Fe Institute, Santa Fe, USA}
\author{Fabio Pasqualetti}
\affiliation{Department of Mechanical Engineering, University of
  California at Riverside, Riverside, USA\\ To whom correspondence should be addressed: \href{mailto:fabiopas@engr.ucr.edu}{fabiopas@engr.ucr.edu}}

\date{\today}

\ \\

\begin{abstract}
Our ability to manipulate the behavior of complex networks depends
  on the design of efficient control algorithms and, critically, on
  the availability of an accurate and tractable model of the network
  dynamics. While the design of control algorithms for network systems
  has seen notable advances in the past few years, knowledge of the
  network dynamics is a ubiquitous assumption that is difficult to
  satisfy in practice, especially when the network topology is large
  and, possibly, time-varying. In this paper we overcome this
  limitation, and develop a data-driven framework to control a complex
  dynamical network optimally and without requiring any knowledge of
  the network dynamics. Our optimal controls are constructed using a
  finite set of experimental data, where the unknown complex network
  is stimulated with arbitrary and possibly random inputs. In addition
  to optimality, we show that our data-driven formulas enjoy favorable
  computational and numerical properties even compared to their
  model-based counterpart. Although our controls are provably correct
  for networks with linear dynamics, we also characterize their
  performance against noisy experimental data and in the presence of
  nonlinear dynamics, as they arise when mitigating cascading failures in power-grid networks and when manipulating neural activity in brain networks.
\end{abstract}

\maketitle


\section{Introduction}
With the development of sensing, processing, and storing capabilities of modern sensors, massive volumes of information-rich data are now rapidly expanding in many physical and engineering domains, ranging from robotics \cite{SL-PP-AK-JI-DQ:18}, to biological \cite{VM:13,TJS-PSC-JAM:14} and economic sciences \cite{LE-JL:14}. Data are often dynamically generated by complex interconnected processes, and encode key information about the structure and operation of these networked phenomena. Examples include temporal recordings of functional activity in the human brain \cite{NBT:13}, phasor measurements of currents and voltages in the power distribution grid \cite{AB:10}, and streams of traffic data in urban transportation networks \cite{YL-YD-WK-ZL-FYW:14}. When first-principle models are not conceivable, costly, or difficult to obtain, this unprecedented availability of data offers a great opportunity for scientists and practitioners to better understand, predict, and, ultimately, control the behavior of real-world complex networks.

Existing works on the controllability of complex networks have focused exclusively on a model-based setting \cite{YYL-JJS-ALB:11,FP-SZ-FB:13q,NB-GB-SZ:17,GY-GT-BB-JS-YL-AB:15,SG-FP-MC-QKT-BYA-AEK-JDM-JMV-MBM-STG-DSB:15,YYL-ALB:16,GL-CA:16}, although, in practice, constructing accurate models of large-scale networks is a challenging, often unfeasible, task \cite{JC-SW:08,SGS-MT:11,MTA-JAM-GL-ALB-YYL:17}. In fact, errors in the network model (i.e., missing or extra links, incorrect link weights) are unavoidable, especially if the network is identified from data (see, e.g., \cite{DA-AC-DK-CM:09,MSH-KJG:10} and Fig.~\ref{fig:error_gramian}(a)).  This uncertainty is particularly important for network controllability, since, as exemplified in Fig.~\ref{fig:error_gramian}(b)-(c), the computation of model-based network controls tends to be unreliable and highly sensitive to model uncertainties, even for moderate size networks~\cite{JS-AEM:13,LZW-YZC-WXW-YCL:17}. It is therefore natural to ask whether network controls can be learned directly from data, and, if so, how well these data-driven control policies perform.

 \begin{figure}[t]
\centering

\

\vspace{2cm}

\includegraphics[width=0.75\linewidth]{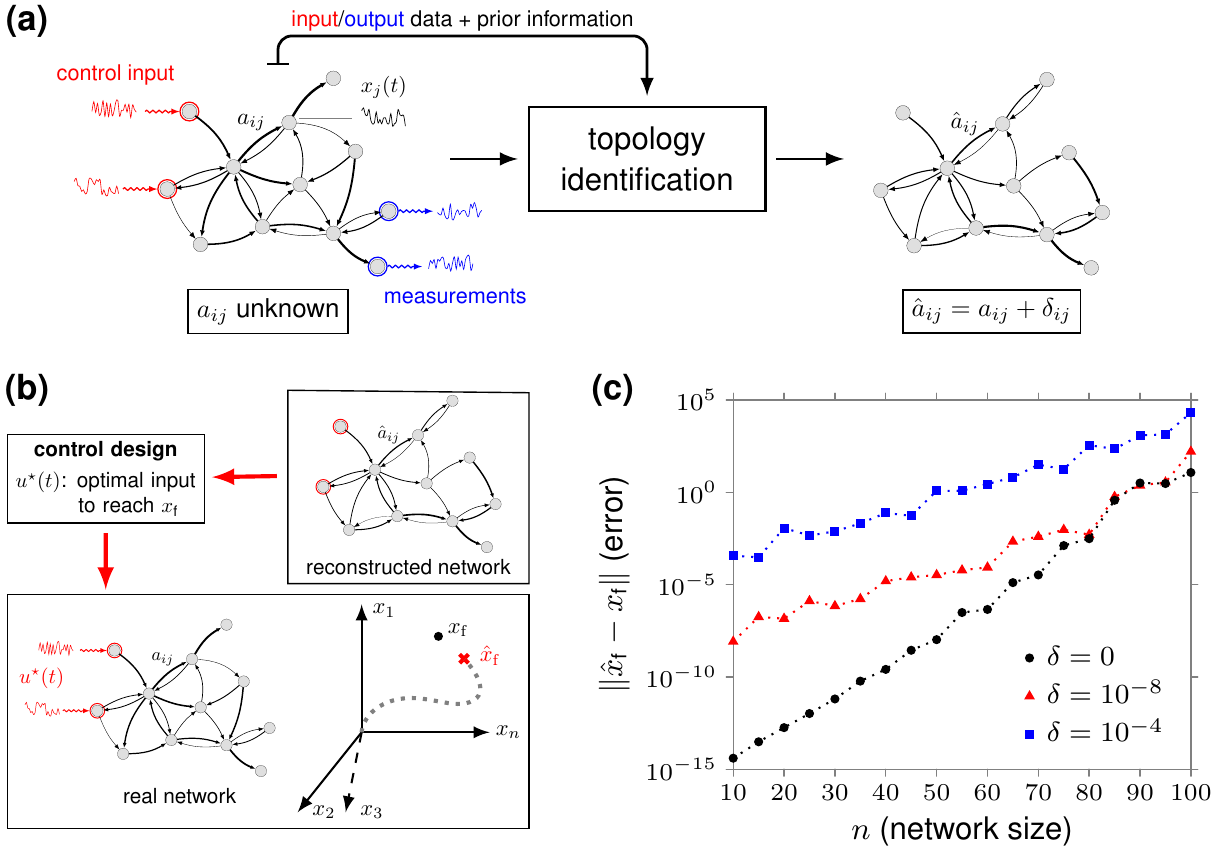} \vspace{1cm}
\caption{\linespread{1}\selectfont{}\textbf{The effect of model uncertainty in the computation of optimal network controls.} Panel \textbf{(a)} shows a schematic of a classic network identification procedure. The reconstructed network is affected by estimation errors~$\delta_{ij}$. Panel \textbf{(b)} illustrates the error in the final state induced by an optimal control design based on the reconstructed network. In panel \textbf{(c)}, we consider minimum-energy controls designed from exact and incorrectly reconstructed networks, and compute the resulting error in the final state as the network size $n$ varies.  We consider connected Erd\"os--R\'enyi networks with edge probability $p=\ln n /n +0.1$, $10$ randomly selected control nodes, control horizon $T=2n$, and a randomly chosen final state $\subscr{x}{f}$.  Each curve represents the average of the error in the final state over 100 random realizations. To mimic errors in the network reconstruction process, we add to each edge of the network a disturbance modeled as an i.i.d.~random variable uniformly distributed in $[-\delta,\delta]$, $\delta>0$. To compute minimum-energy control inputs, we use the classic Gramian-based formula and standard LAPACK linear algebra routines~(see~\methods). 
} 
\label{fig:error_gramian}
\end{figure}

Data-driven control of dynamical systems has attracted increasing interest over the last few years, triggered by recent advances and
successes in machine learning and artificial intelligence \cite{SL-CF-TD-PA:16,DS-et-al:17}. The classic (indirect) approach to learn controls from data is to use a sequential system identification and control design procedure. That is, one first identifies a model of the system from the available data, and then computes the desired controls using the estimated model \cite{MG:05}. However, identification algorithms are sometimes inaccurate and time-consuming, and several direct data-driven methods have been proposed to bypass the identification step \cite[Ch. III.10]{SLB-JNK:19}. These include, among others, (model-free) reinforcement learning \cite{FLL-DV-KGV:12,BR:18}, iterative learning control \cite{DAB-MT-AGA:06}, adaptive and self-tuning control \cite{KJA-BW:73}, and behavior-based~methods~\cite{IM-PR:08,CDP-PT:19}.

The above techniques differ in the data generation procedure, class of system dynamics considered, and control objectives. In classic reinforcement learning settings, data are generated online and updated under the guidance of a policy evaluator or reward
simulator, which in many applications is represented by an offline-trained (deep) neural network \cite{DPB-JNT:96}. Iterative learning control is used to refine and optimize repetitive control tasks: data are recorded online during the execution of a task repeated multiple times, and employed to improve tracking accuracy from trial to trial. In adaptive control, the structure of the controller is fixed and a few control parameters are optimized using data collected on the fly. A widely known example is the auto-tuning of PID controllers \cite{KJA-TH:95}. Behavior-based techniques exploit a trajectory-based (or behavioral) representation of the system, and data that typically consist of a single, noiseless, and sufficiently long input-output system trajectory \cite{CDP-PT:19}. Each of the above data-driven approaches has its own limitations and merits, which strongly depend on the intended application area. However, a common feature of all these approaches is that they are tailored to or have been employed for closed-loop control tasks, such as stabilization or tracking, and not for finite-time point-to-point~control~tasks.

In this paper, we address the problem of learning from data point-to-point optimal controls for complex dynamical networks. Precisely, following recent literature on the controllability of complex networks \cite{JG-YYL-RMD-ALB:14,IK-AS-FS:17b}, we focus on control policies that optimally steer the state of (a subset of) network nodes from a given initial value to a desired final one within a finite time horizon. To derive analytic, interpretable results that capture the role of the network structure, we consider networks governed by linear dynamics, quadratic cost functions, and data consisting of a set of control experiments recorded offline. Importantly, experimental data are not required to be optimal, and can even be generated through random control experiments. In this setting, we establish closed-form expressions of optimal data-driven control policies to reach a desired target state and, in the case of noiseless data, characterize the minimum number of experiments needed to exactly reconstruct optimal control inputs. Further, we introduce suboptimal yet computationally simple data-driven expressions, and discuss the numerical and computational
advantages of using our data-driven approach when compared to the classic model-based one. Finally, we illustrate with different numerical studies how our framework can be applied to restore the correct operation of power-grid networks after a fault, and to characterize the controllability properties of functional brain networks.

While the focus of this paper is on designing optimal control inputs, the expressions derived in this work also provide an
alternative, computationally reliable, and efficient way of analyzing the controllability properties of large network systems. This
constitutes a significant contribution to the extensive literature on the model-based analysis of network controllability, where the
limitations imposed by commonly used Gramian-based techniques limit the investigation to small and well-structured network topologies \cite{JS-AEM:13, LZW-YZC-WXW-YCL:17}.

\section{Results}

\subsection{Network dynamics and optimal point-to-point control} We consider networks governed by linear time-invariant dynamics
\begin{linenomath*}
\begin{equation}\label{eq:sys}
\begin{aligned}
	x(t+1) &= A x(t) + B u(t),  \\
	y(t) &= C x(t), 
\end{aligned}
\end{equation}
\end{linenomath*}
where $x(t)\in\real^{n}$, $u(t)\in\real^{m}$, and $y(t)\in\real^{p}$ denote, respectively, the state, input, and output of the network at
time~$t$. The matrix $A\in\real^{n\times n}$ describes the (directed and weighted) adjacency matrix of the network, and the  matrices $B\in\real^{n\times m}$ and $C\in\real^{p\times n}$, respectively, are typically chosen to single out prescribed sets of input and~output~nodes~of~the~network. 

In this work, we are interested in designing open-loop control policies that steer the network output $y(t)$ from an initial value
$y(0)=y_{0}$ to a desired one $y(T)=\subscr{y}{f}$ in $T$ steps. If $\subscr{y}{f}$ is output controllable \cite{TK:80,JG-YYL-RMD-ALB:14} (a standing assumption in this paper), then the latter problem admits a solution and, in fact, there are many ways to accomplish such a control task. Here, we assume that the network is initially relaxed ($x(0) = 0$), and we seek the control input
$\smash{u^{\star}_{0:T-1} = [u^{\star}(T-1)^{\transpose}\cdots\,  u^{\star}(0)^{\transpose}]^{\transpose}}$ that drives the output of
the network to $\subscr{y}{f}$ in $T$ steps and, at the same time, minimizes a prescribed quadratic combination of the control effort and locality of the controlled trajectories.

Mathematically, we study and solve the following constrained minimization~problem: 
\begin{linenomath*} 
\begin{equation}\label{eq:optimal_control} \begin{aligned}
	u^{\star}_{0:T-1} = \arg\min_{u_{0:T-1}}\ \ &  y_{1:T-1}^{\transpose}\, Q\, y_{1:T-1} + u_{0:T-1}^{\transpose}\, R\, u_{0:T-1} \\
	&\quad  \text{s.t. }\  \eqref{eq:sys}\ \text{ and }\  y_T = \subscr{y}{f},
\end{aligned}
\end{equation}
\end{linenomath*}
where $Q\succeq 0$ and $R\succ 0$ are tunable matrices\footnote{We let $A\succ (\succeq)\, 0$ denote a positive definite (semi-definite) matrix, and $A^{\transpose}$~the~transpose~of~$A$.} that penalize output deviation and input usage, respectively, and subscript $\cdot_{t_{1}:t_{2}}$ denotes the vector containing the samples of a trajectory in the time window $[t_{1},t_{2}]$, $t_{1}\le t_{2}$ (if $t_1=t_2$, we simply write $\cdot_{t_1}$). If $Q=0$ and $R=I$, then $u^{\star}_{0:T-1}$ coincides with the minimum-energy control to reach $\subscr{y}{f}$ in $T$ steps \cite{TK:80}. 

Equation \eqref{eq:optimal_control} admits a closed-form solution whose computation requires the exact knowledge of the network matrix $A$ and suffers from numerical instabilities (\methods). In the following section, we address this limitation by deriving model-free and reliable expressions of $u^{\star}_{0:T-1}$ that solely rely on experimental data collected during the~network~operation.

\begin{figure*}[t]  
\ 

\vspace{0.5cm}

\centering   \includegraphics[width=0.925\linewidth]{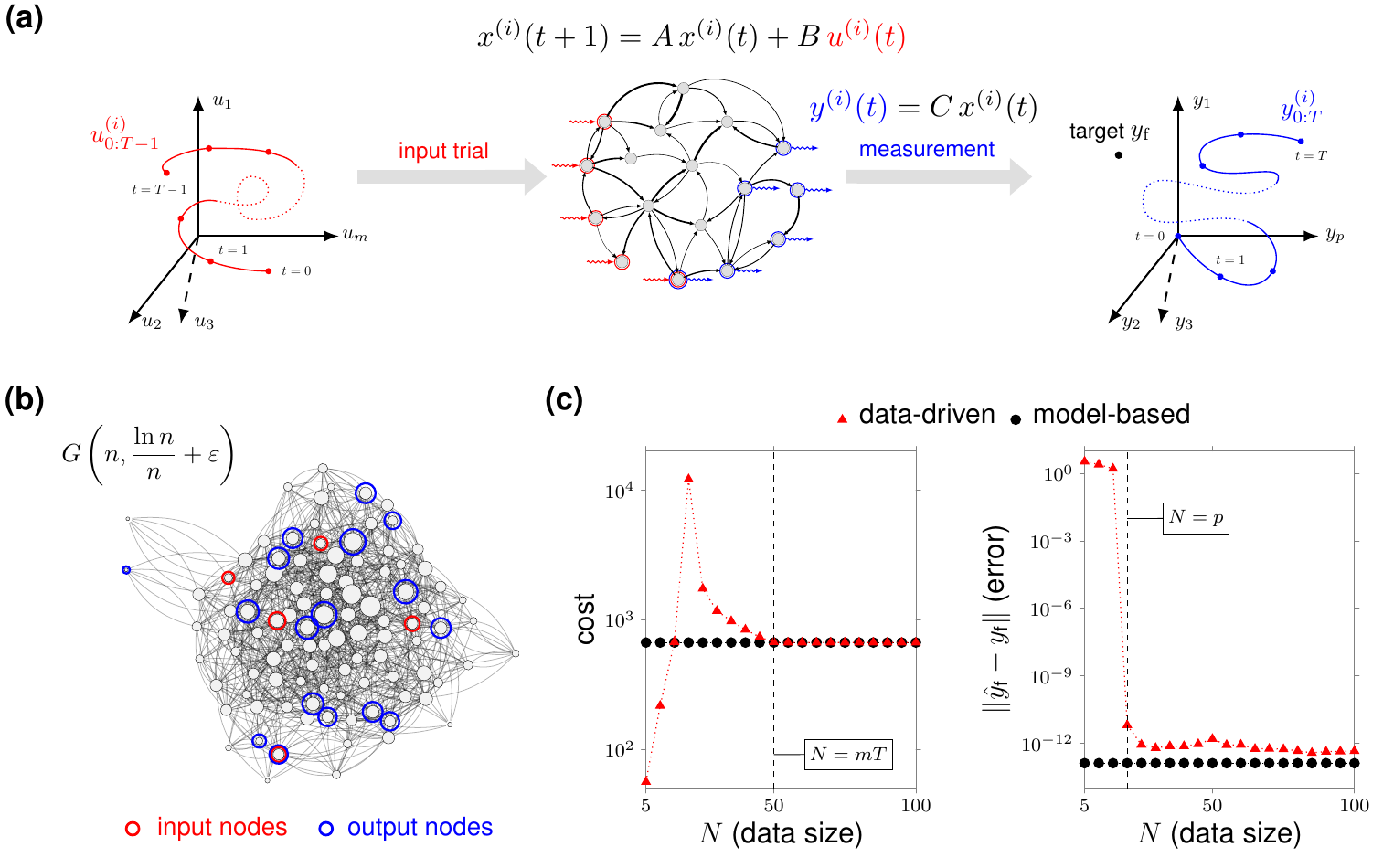}  \vspace{0.5cm} \caption{\linespread{1}\selectfont{}\textbf{Experimental setup and optimal data-driven network controls.} Panel \textbf{(a)} illustrates the data collection process. With reference to the $i$-th control experiment, a $T$-step input sequence $\smash{u_{0:T}^{(i)}}$ excites the network dynamics in \eqref{eq:sys}, and the time samples of the resulting output trajectory $\smash{y_{0:T}^{(i)}}$ are recorded. The input trajectory $\smash{u_{0:T}^{(i)}}$ may be generated randomly, so that the final output $\smash{y^{(i)}(T)}$ does not normally coincide with the desired target output $y_{\mathrm{f}}$. Red nodes denote the control or input nodes (forming matrix $B$) and the blue nodes denote the measured or output nodes (forming matrix $C$). Panel \textbf{(b)} shows a realization of the Erd\"os--R\'enyi graph model $G(n,p_{\text{edge}})$ used in our examples, where $n$ is the number of nodes, $p_{\text{edge}}$ is the edge probability, $m$ is the number of input nodes (red nodes), and $p$ the number of output nodes (blue nodes). We set the edge probability to $p=\ln n/n +\varepsilon$, $\varepsilon=0.05$, to ensure connectedness with high probability, and normalize the resulting adjacency matrix by a factor $\smash{\sqrt{n}}$. Panel \textbf{(c)} shows the value of the cost function (\emph{left}) and the error in the final state (\emph{right}) for the data-driven input \eqref{eq:dd} and the model-based control for a randomly chosen target $y_{\text{f}}$, as a function of the number of data points. We choose $Q=R=I$, $n=100$, $T=10$, $m=5$, and $p=20$, and consider Erd\"os--R\'enyi networks as in panel \textbf{(b)}. In all simulations the entries of the input data matrix $U$ are normal i.i.d.~random variables, and the input and output nodes are randomly selected. Target controllability is always ensured for all choices of input nodes by adding self-loops and edges that guarantee strong connectivity when needed. All curves represent the average over 500 realizations of data, networks, and input/output nodes. For additional computational details, see \methods.}
\label{fig:experiments}
\end{figure*}

\subsection{Learning optimal controls from non-optimal data} 
We assume that the network matrix $A$ is unknown and that $N$ control experiments have been performed with the dynamical network in \eqref{eq:sys}. The $i$-th experiment consists of generating and applying the input sequence $\smash{u_{0:T-1}^{(i)}}$, and measuring the resulting output trajectory $\smash{y_{0:T}^{(i)}}$ (Fig.~\ref{fig:experiments}(a)). Here, as in, e.g., \cite{SD-HM-NM-BR-ST:18}, we consider episodic experiments where the network state is reset to zero before running a new trial, and refer to the \SI\ for an extension to the non-episodic setting and to the case of episodic experiments with non-zero initial state resets. We let $U_{0:T-1}$, $Y_{1:T-1}$, and $Y_{T}$ denote the matrices containing, respectively, the experimental inputs, the output measurements in the time interval $[1,T-1]$, and the output measurements at time $T$. Namely, 
\begin{linenomath*} 
\begin{equation} 
\begin{aligned}\label{eq:data}
	U_{0:T-1} &= \left[u_{0:T-1}^{(1)}\ \ \cdots\ \ u_{0:T-1}^{(N)}\right], \\
	Y_{1:T-1} &= \left[y_{1:T-1}^{(1)} \ \ \cdots\ \ y_{1:T-1}^{(N)}\right], \\ 
	Y_{T} &= \Big[y^{(1)}_T\ \ \cdots\ \ y^{(N)}_T\Big].
\end{aligned}
\end{equation}
\end{linenomath*}
An important aspect of our analysis is that we do not require the input experiments to be optimal, in the sense of \eqref{eq:optimal_control}, nor do we investigate the problem of experiment design, i.e., generating data that are ``informative'' for
our problem. In our setting, data are given, and these are generated from arbitrary, possibly random, or carefully chosen
experiments.

By relying on the data matrices in \eqref{eq:data}, we derive the following data-driven solution to the minimization problem in \eqref{eq:optimal_control} (see the \SI):
\begin{linenomath*}
\begin{equation}\label{eq:dd}
	\hat{u}_{0:T-1} = U_{0:T-1}(I-K_{Y_{T}}(LK_{Y_{T}})^{\dag}L)Y_{T}^{\dag}\,\subscr{y}{f},
\end{equation}
\end{linenomath*}
where $L$ is any matrix satisfying $L^{\transpose}L=Y_{1:T-1}^{\transpose}QY_{1:T-1} + U_{0:T-1}^{\transpose}RU_{0:T-1}$, $K_{Y_{T}}$ denotes a matrix whose columns form a basis of the kernel of $Y_{T}$, and the superscript symbol $\cdot^{\dag}$ stands for the Moore--Penrose pseudoinverse operation \cite{AB-TNEG:03}.  
 
\subsubsection{Minimum number of data to learn optimal controls} 
Finite data suffices to exactly reconstruct the optimal control input via the data-driven expression in \eqref{eq:dd} (see the \SI). In Fig.~\ref{fig:experiments}(c), we illustrate this fact for the class of Erd\"os--R\'enyi networks of Fig.~\ref{fig:experiments}(b). Specifically, the data-driven input $\hat{u}_{0:T-1}$ equals the optimal one $u^{\star}_{0:T-1}$ (for any target $\subscr{y}{f}$) if the data matrices in \eqref{eq:data} contain $mT$ linearly independent experiments; that is, if $U_{0:T-1}$ is full row rank (Fig.~\ref{fig:experiments}(c), \emph{left}). We stress that linear independence of the control experiments is a mild condition that is normally satisfied when the experiments are generated randomly. Further, if the number of independent trials is smaller than $mT$ but larger than or equal to $p$, the data-driven control $\hat{u}_{0:T-1}$ still correctly steers the network output to $\subscr{y}{f}$ in $T$ steps (Fig.~\ref{fig:experiments}(c), \emph{right}), but with a cost that is typically larger than the optimal one. In this case, $\hat{u}_{0:T-1}$ is a suboptimal solution to \eqref{eq:optimal_control}, which becomes optimal (for any $\subscr{y}{f}$) if the collected data contain $p$ independent trials that are~optimal~as~well.

\begin{figure*}[t] 
\
\vspace{0.75cm}

 \centering 
 
 \includegraphics[width=1\linewidth]{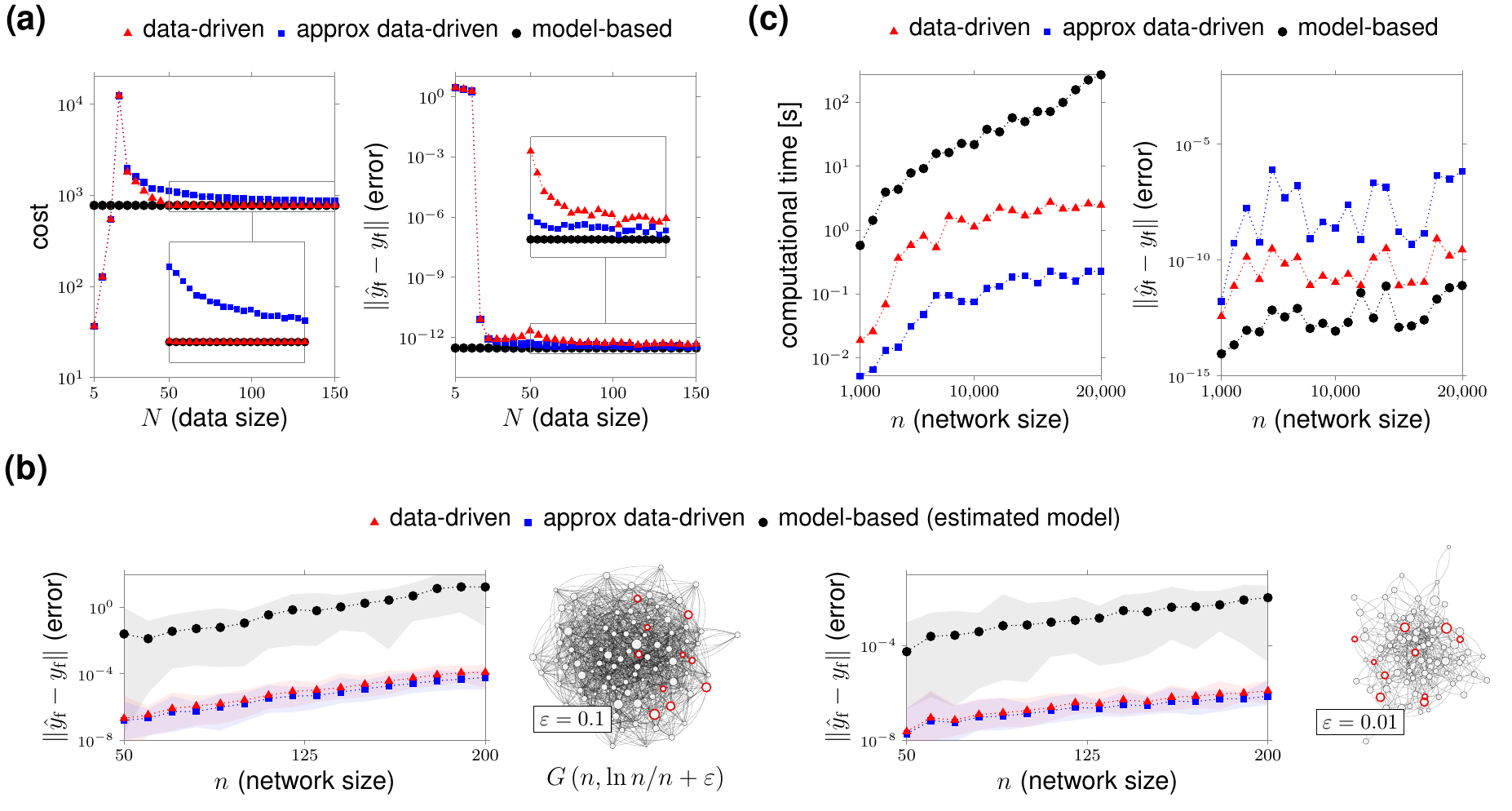}  \vspace{0.25cm} \caption{\linespread{1}\selectfont{}\textbf{Performance of minimum-energy data-driven network controls.} Panel \textbf{(a)} shows the value of the cost function (\emph{left}) and the error in the final state (\emph{right}) for the minimum-energy data-driven control inputs \eqref{eq:dd-min-en} and \eqref{eq:dd-min-en-approx}, and the model-based one as a function of the number of data $N$. We consider a randomly chosen target $y_{\text{f}}$, Erd\"os--R\'enyi networks as in Fig.~\ref{fig:experiments}\textbf{(b)} with $\varepsilon=0.05$, and parameters $n=100$, $T=10$, $m=5$, $p=20$. In Panel \textbf{(b)}, we compare the error in the final state of the data-driven approach (\eqref{eq:dd-min-en} and \eqref{eq:dd-min-en-approx}), and the classic two-step approach (i.e., network identification followed by model-based control design) for a randomly chosen target $y_{\text{f}}$, as a function of the network size $n$. We use the subspace-based identification procedure described in \methods, Erd\"os--R\'enyi networks as in Fig.~\ref{fig:experiments}\textbf{(b)} with two different edge densities determined through the parameter $\varepsilon$, and parameters $T=40$, \smash{$m=\lfloor n/10\rfloor$}, $p=n$, $N=mT+10$. The curves in panels \textbf{(a)} and \textbf{(b)} represent the average over 500 realizations of data, networks, and input/output nodes, and the light-colored regions in panel \textbf{(b)} contain the values of all realizations. Panel \textbf{(c)}, \emph{left}, compares the time needed to compute the optimal controls via data-driven and model-based strategies as a function of the network size, for a randomly chosen target $y_{\text{f}}$ and one realization of the Erd\"os--R\'enyi network model and data. Panel \textbf{(c)}, \emph{right}, shows the errors in the final state. We use the following parameters: $\varepsilon=0.05$, \smash{$m=\lfloor n/100\rfloor$}, \smash{$p=\lfloor n/50\rfloor$}, $\smash{T=50}$, and $N=mT+100$. In all simulations the entries of the input data matrix $U_{0:T-1}$ are normal i.i.d.~variables, and the input and output nodes are randomly selected. Further, target controllability is always ensured for all choices of input nodes by adding self-loops and edges that guarantee strong connectivity when needed. For additional computational details, see~\methods.  }
\label{fig:experiments2}
\end{figure*}

\subsubsection{Data-driven minimum-energy control} By letting $Q=0$ and $R=I$ in \eqref{eq:dd}, we recover a data-driven expression for the $T$-step minimum-energy control to reach $\subscr{y}{f}$. We remark that the family of minimum-energy controls has been extensively employed to characterize the fundamental capabilities and limitations of controlling networks, e.g., see \cite{FP-SZ-FB:13q,GY-GT-BB-JS-YL-AB:15,GL-CA:16}. After some algebraic manipulations, the data-driven minimum-energy control input can be compactly rewritten as (see the \SI)
\begin{linenomath*}
\begin{equation}\label{eq:dd-min-en}
	\hat{u}_{0:T-1} 
	= (Y_{T}U_{0:T-1}^{\dag})^{\dag}\,\subscr{y}{f}.
\end{equation}
\end{linenomath*}
The latter expression relies on the final output measurements only (matrix $Y_{T}$) and, thus, it does not exploit the full output data
(matrix $Y_{1:T-1}$).  Equation \eqref{eq:dd-min-en} can be further approximated~as
\begin{linenomath*}
\begin{equation}\label{eq:dd-min-en-approx}
	\hat{u}_{0:T-1} \approx U_{0:T-1} Y_{T}^{\dag}\,\subscr{y}{f}.
\end{equation}
\end{linenomath*}
This is a simple, suboptimal data-based control sequence that correctly steers the network to $\subscr{y}{f}$ in $T$ steps, as long as $p$ independent data are available. Further, and more importantly, when the input data samples are drawn randomly and independently from a distribution with zero mean and finite variance, \eqref{eq:dd-min-en-approx} converges to the minimum-energy control in the limit of infinite data (see the \SI). 

Fig.~\ref{fig:experiments2}(a) compares the performance (in terms of control effort and error in the final state) of the two data-driven expressions \eqref{eq:dd-min-en} and \eqref{eq:dd-min-en-approx}, and the model-based control as a function of the data size $N$. While the data-driven control in \eqref{eq:dd-min-en} becomes optimal for a finite number of data (precisely, for $N=mT$), the approximate expression \eqref{eq:dd-min-en-approx} tends to the optimal control only asymptotically in the number of data (Fig.~\ref{fig:experiments2}(a), \emph{left}). In both cases, the error in the final state goes to zero after collecting $N=p$ data (Fig.~\ref{fig:experiments2}(a), \emph{right}). For the approximate control \eqref{eq:dd-min-en-approx}, we also establish upper bounds on the size of the dataset to get a prescribed deviation from the optimal cost in the case of Gaussian noise. Our non-asymptotic analysis indicates that this deviation is proportional to the worst-case control energy required to reach a unit-norm target. This, in turn, implies that networks that are ``easy'' to control require fewer trials to attain a prescribed approximation error~(see the \SI).

\subsubsection{Numerical and computational benefits of data-driven
  controls} By relying on the same set of experimental data, in
Fig.~\ref{fig:experiments2}(b), we compare the numerical accuracy, as
measured by the error in the final state, of the data-driven controls
 \eqref{eq:dd-min-en} and \eqref{eq:dd-min-en-approx} and the
minimum-energy control computed via a standard two-step approach
comprising a network identification step followed by model-based
control design. First, we point out that if some nodes of the network are not accessible $(C\ne I)$ and no prior information about the network structure is available, then it is impossible to exactly reconstruct the network matrix $A$ using (any number of) data \cite{PPE-VC-SW:13}. In contrast, the computation of minimum-energy inputs is always feasible via our data-driven expression, provided that enough data are collected. We thus focus on the case in which all nodes can be
accessed ($C=I$). We consider Erd\"os--R\'enyi networks with $n$ nodes as in
Fig.~\ref{fig:experiments}(b) and we select $m= \lfloor n/10  \rfloor$ control nodes (forming matrix $B$). To reconstruct the network matrices $A$ and
$B$, the subspace-based identification technique described in
\methods. Data-driven strategies significantly outperform the standard
sequential approach for both dense (Fig.~\ref{fig:experiments2}(b),
\emph{top}) and sparse topologies (Fig.~\ref{fig:experiments2}(b),
\emph{bottom}). This poor performance of the standard approach is somehow expected because, independently of the
network identification procedure, the standard two-step approach
requires a number of operations larger than those required by the
data-driven approach, resulting in a potentially higher sensitivity to
round-off errors.  Also, it is interesting to note that the
data-driven approach is especially effective for large, dense networks
for which the standard approach leads to errors of considerable
magnitude (up to approximately $10^{2}$).

A further advantage in using data-driven controls over model-based ones arises when dealing with massive networks featuring a small fraction of input and output nodes. Specifically, in Fig.~\ref{fig:experiments2}(c) we plot the time needed to numerically compute the data-driven and model-based controls as a function of the size of the network. We focus on Erd\"os--R\'enyi networks as in Fig.~\ref{fig:experiments}(b) of dimension $n\ge 1000$ with $\lfloor n/100 \rfloor$ input and output nodes and a control horizon $T=50$. The model-based control input requires the computation of the first $T-1$ powers of $A$ (\methods). The computation of the data-driven expressions \eqref{eq:dd-min-en} and \eqref{eq:dd-min-en-approx} involves, instead, linear-algebraic operations on two matrices ($U_{0:T-1}$ and $Y_{T}$) that are typically smaller than $A$ when $n$ is very large  (precisely, when $T<n/m$ and $N<n$). 
Thus, the computation of the control input via the data-driven
approach is normally faster than the classic model-based computation
(Fig.~\ref{fig:experiments2}(c), \emph{left}). In particular, the
data-driven control \eqref{eq:dd-min-en-approx}, although
suboptimal, yields the most favorable performance due to its
particularly simple expression. Finally, we note that the error in the final state committed by the data-driven controls is always upper bounded by $10^{-5}$ and thus it has a negligible effect on the control accuracy (Fig.~\ref{fig:experiments2}(c),~\emph{right}).

\subsubsection{Data-driven controls with noisy data}
The analysis so far has focused on noiseless data. A natural question is how the data-driven controls behave in the case of noisy data. If the noise is unknown but small in magnitude, then the established data-driven expressions will deviate slightly from the correct values (see the \SI). However, if some prior information on the noise is known, this information can be exploited to return more accurate control expressions. A particularly relevant case is when data are corrupted by additive i.i.d.~noise with zero mean and known variance.\footnote{The different types of noise are assumed to be zero-mean to simplify the exposition. With slight modifications, non zero-mean noise could also be accommodated by our approach.}
Namely, the available data read as
\begin{linenomath*} 
\begin{equation} 
\begin{aligned}\label{eq:noisy-data}
	U_{0:T-1} &= \bar U_{0:T-1} +\Delta_{U}, \\
	Y_{1:T-1} &= \bar Y_{1:T-1} + \Delta_{Y}, \\ 
	Y_{T} &= \bar Y_{T} + \Delta_{Y_{T}},
\end{aligned}
\end{equation}
\end{linenomath*}
where $\bar U_{0:T-1}$, $\bar Y_{1:T-1}$, $\bar Y_{T}$ denote the ground truth values and $\Delta_{U}$, $\Delta_{Y}$, and $\Delta_{Y_{T}}$ are random matrices with i.i.d.~entries with zero mean and variance $\sigma_{U}^{2}$, $\sigma_{Y}^{2}$, and $\sigma_{Y_{T}}^{2}$, respectively. Note that the noise terms $\Delta_{Y}$ and $\Delta_{Y_{T}}$ may also include the contribution of process noise acting on the network dynamics. In this setting, it can be shown that the data-driven control \eqref{eq:dd} and the data-driven minimum-energy controls \eqref{eq:dd-min-en} and \eqref{eq:dd-min-en-approx} are typically not consistent; that is, they do not converge to the true control inputs as the data size tends to infinity (see the \SI\ for a concrete example). However, by suitably modifying these expressions, it is possible to recover asymptotically correct data-driven formulas (\SI). The key idea is to add correction terms that compensate for the noise variance arising from the pseudoinverse operations. In particular, the asymptotically correct version of the data-driven controls \eqref{eq:dd-min-en} and \eqref{eq:dd-min-en-approx} read, respectively, as
\begin{linenomath*}
\begin{eqnarray} 
	\hat u'_{0:T-1} &=&
                          (Y_{T}U_{0:T-1}^{\transpose}(U_{0:T-1}U_{0:T-1}^{\transpose}-
                          N\sigma_{U}^{2}I)^{\dag})^{\dag}\subscr{y}{f}
                          , \label{eq:dd-min-energy-corr} \\
	\hat{u}''_{0:T-1} &=& U_{0:T-1} Y_{T}^{\transpose}(Y_{T}Y_{T}^{\transpose}-N\sigma^{2}_{Y_{T}}I)^{\dag}\,\subscr{y}{f}, \label{eq:dd-min-en-approx-noise}
\end{eqnarray}
\end{linenomath*}
where we used the fact that $X^{\dag}=X^{\transpose}(XX^{\transpose})^{\dag}$ for any matrix $X$ \cite{AB-TNEG:03}, and $N\sigma_{U}^{2}I$ and $N\sigma_{Y_{T}}^{2}I$ represent the noise-dependent correction terms. Note, in particular, that if the noise corrupts the output data $Y_{T}$ only, then \eqref{eq:dd-min-energy-corr} coincides with the original data-driven control \eqref{eq:dd-min-en}, so that no correction is needed. Similarly, if the noise corrupts the input data $U_{T}$ only, then \eqref{eq:dd-min-en-approx-noise} coincides with the data-driven control \eqref{eq:dd-min-en-approx}.

\subsection{Applications}

To demonstrate the potential relevance and applicability of the data-driven framework presented thus far, we present two applications of our data-driven control formulas. 

\subsubsection{Data-driven fault recovery in power-grid networks}

\begin{figure*}[t] \centering
  \includegraphics[width=1\linewidth]{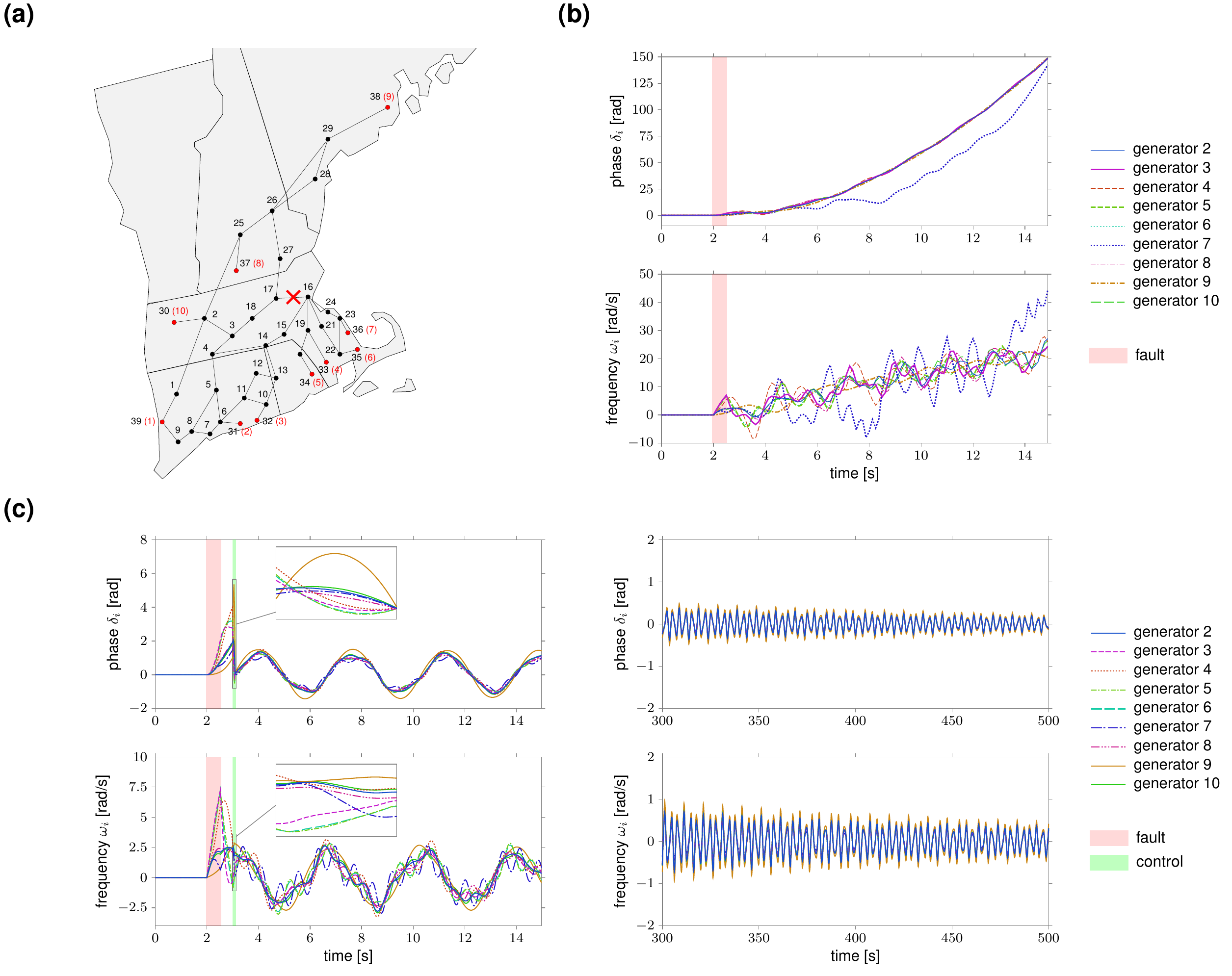}
  \caption{\linespread{1}\selectfont{}\ \ \textbf{Data-driven fault recovery in the New England power-grid network.} Panel \textbf{(a)} depicts the 39-node New England power-grid network. The black nodes $\{1,\dots, 29\}$ represent load nodes, while the red  nodes $\{30,\dots, 39\}$ are power generators. The generators are labelled according to the numbers in the red brackets. The red cross denotes the location of the fault. Panel \textbf{(b)} plots the behavior of the phases and frequencies of generators $\{2,\dots, 10\}$ after the occurrence of the fault. The onset time of the fault is $t=2$s and the fault duration is $0.5$s (red area in the plots). At time $t=2.5$s the fault is cleared. The phase and frequency of generator $1$ (not shown) are fixed to a constant (see \methods). The left plots of panel \textbf{(c)} show the behavior of the phases and frequencies of generators $\{2,\dots, 10\}$ after the application of the data-driven control input \eqref{eq:dd}. The duration of the control action is $0.1$s (green area in the plots) which correspond to a control horizon $T=400$ for the discretized network dynamics with sampling period $T_{s}=2.5\times 10^{-4}$s. For the computation of the control input, we employ $N=4000$ experimental data collected offline by perturbing the state of the generators locally around its steady-state value (see \methods). We use weighting matrices $R=I$ and $Q=\varepsilon I$, where $\varepsilon$ is set to a small non-zero constant ($\varepsilon = 0.01$) to guarantee both limited control effort and locality of the controlled trajectories. The insets illustrate the behavior of the phases and frequencies during the application of the control. The right plots of panel \textbf{(c)} show the asymptotic behavior of the phases and frequencies of generators $\{2,\dots, 10\}$ after the application of the data-driven control \eqref{eq:dd}.
}
\label{fig:grid}
\end{figure*}

We address the problem of restoring the normal operation of a
  power-grid network after the occurrence of a fault which
  desynchronizes part of the grid. If not mitigated in a timely manner, such
  desynchronization instabilities may trigger cascading failures that
  can ultimately cause major blackouts and grid disruptions
  \cite{YS-IM-TH:11,SPC-LWK-AEM:13,JWSP-FD-FB:16}. In our case study,
  we consider a line fault in the New England power grid network comprising 39 nodes
  (29 load nodes and 10 generator nodes), as depicted in
  Fig.~\ref{fig:grid}(a), and we compute an optimal point-to-point control from data to recover the correct operation of the grid. A similar problem is solved in \cite{SPC-LWK-AEM:13} using a more sophisticated control strategy which requires knowledge of the network dynamics. As in \cite{YS-IM-TH:11,SPC-LWK-AEM:13}, we
  assume that the phase $\delta_{i}$ and the (angular) frequency
  $\omega_{i}$ of each generator $i$ obey the swing equation dynamics
  with the parameters given in \cite{YS-IM-TH:11} (except for
  generator 1 whose phase and frequency are fixed to a constant, cf.
  \methods). Initially, each generator operates at a locally stable
  steady-state condition determined by the power flow equations. At
  time $t=2$s, a three-phase fault occurs in the transmission line
  connecting nodes 16 and 17. After $0.5$s the fault is cleared;
  however the generators have lost synchrony and deviate from their
  steady-state values (Fig.~\ref{fig:grid}(b)). To recover the normal
  behavior of the grid, $0.5$s after the clearance of the fault, we
  apply a short, optimal control input to the frequency of the
  generators to steer the state (phase and frequency) of the
  generators back to its steady-state value. The input is computed
  from data via \eqref{eq:dd} using $N=4000$ input/state experiments
  collected by locally perturbing the state of the generators around
  its normal operation point (see also \methods). We consider data
  sampled with period $T_{s}=2.5\times 10^{-4}$s, and set the control
  horizon to $T=400$ time samples (corresponding to $0.1$s), $R=I$,
  and $Q=\varepsilon I$ with $\varepsilon = 0.01$ to enforce locality
  of the controlled trajectories. As shown in Fig.~\ref{fig:grid}(c),
  the data-driven input drives the state of the generators to a point
  close enough to the starting synchronous solution (\emph{left},
  inset) so as to asymptotically recover the correct operation
  of the grid (\emph{right}). Notably, as previously discussed,
  the computation of the control input requires only pre-collected
  data, is numerically efficient, and optimal (for the linearized dynamics).
  More generally, this numerical study
  suggests that the data-driven strategy \eqref{eq:dd} could represent
  a simple, viable, and computationally-efficient approach to control
  complex non-linear networks around an operating point.

\begin{figure*} \centering 
  \includegraphics[width=1\linewidth]{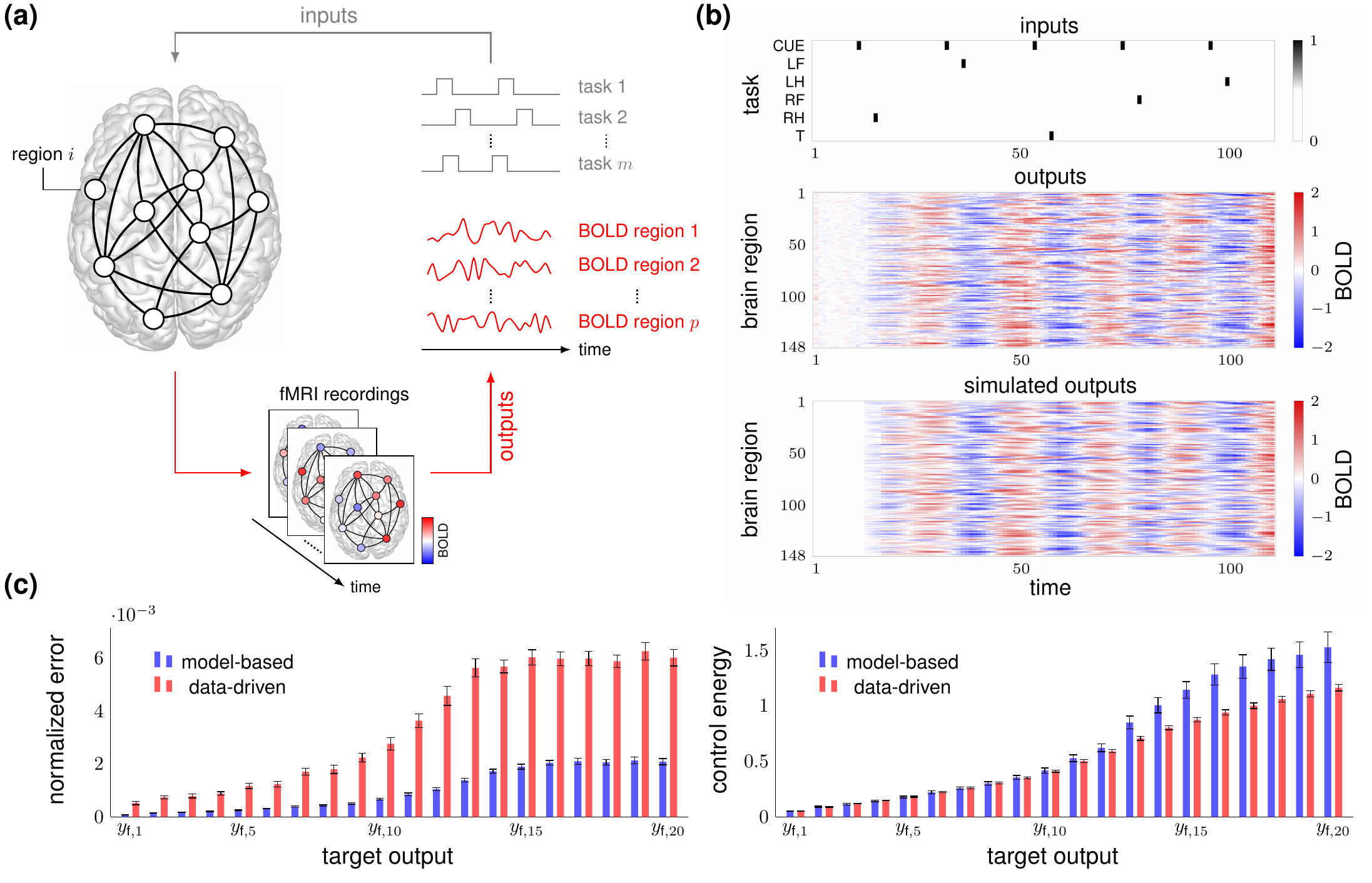}
  \caption{\linespread{1}\selectfont{}\textbf{Data-driven control of functional brain networks.}
    Panel \textbf{(a)} provides a schematic of the experimental setup.
    A set of external stimuli represented by $m$ different task
    commands induce brain activity. Functional magnetic resonance
    (fMRI) blood oxygen level dependent (BOLD) signals are measured
    and recorded at different times and converted into $p$ time
    series, one for each brain region. The top and center heatmaps of
    panel \textbf{(b)} show the inputs and outputs, respectively, for
    the first 110 measurements of one subject of the HCP dataset. The
    inputs are divided into $m=6$ channels corresponding to different
    task conditions, i.e., CUE (a visual cue preceding the occurrence
    of other task conditions), LF (squeeze left toe), LH (tap left
    fingers), RF (squeeze right toe), RH (tap right finger), and T
    (move tongue). As in \cite{CO-DSB-VMP:18}, each input is a binary
    0-1 signal taking the value 1 when the corresponding task
    condition is issued and 0 otherwise. The outputs represent the
    BOLD signals of the $p=148$ brain regions obtained from and
    enumerated according to the Destrieux 2009 atlas
    \cite{CD-BF-AD-EH:10}. These outputs have been minimally
    pre-processed following standard techniques \cite{MFG-et-al:13} as
    detailed in \methods. The bottom heatmap of panel \textbf{(b)}
    displays the simulated outputs obtained by exciting the
    approximate low-dimensional linear model of \cite{CO-DSB-VMP:18}
    with the input sequence of the top plot. In panel \textbf{(c)}, we
    compare the performance of the data-driven and model-based
    strategy, assuming that the dynamics obey the above-mentioned
    approximate linear model. We set the control horizon to $T=100$
    and generate the data matrices by sliding a time window of size
    $T$ across the data samples. The target state $y_{\text{f},i}$ is
    the eigenvector associated with the $i$-th eigenvalue of the
    empirical Gramian matrix
    \smash{$\hat{\mathcal{W}}_{T}=\hat{\mathcal{C}}_{T}^{\transpose}\hat{\mathcal{C}}_{T}$},
    where \smash{$\hat{\mc C}_{T} = Y_{T}U_{0:T-1}^{\dagger}$}. The
    top plot shows the error to reach the targets
    \smash{$\{y_{\text{f},i}\}_{i=1}^{20}$} using the data-driven
    minimum-energy input in \eqref{eq:dd-min-en} and the model-based
    one. The bottom plot shows the norm of the two inputs. The colored
    bars denote the mean over 100 unrelated subjects and the error
    bars represent the 95\% confidence intervals around the mean. }
\label{fig:fMRI}
\end{figure*}

\subsubsection{Controlling functional brain networks via fMRI snapshots} 
We investigate the problem of generating prescribed patterns of activity in functional brain networks directly from task-based functional magnetic resonance imaging (task-fMRI) time series. Specifically, we examine a dataset of task-based fMRI experiments related to motor activity extracted from the Human Connectome Project (HCP) \cite{DCVE-et-al:13} (see Fig.~\ref{fig:fMRI}(a)). In these experiments, participants are presented with visual cues that ask them to execute specific motor tasks; namely, tap their left or right fingers, squeeze their left or right toes, and move their tongue. We consider a set of $m=6$ input channels associated with different task-related stimuli; that is, the motor tasks' stimuli and the visual cue preceding them. As in \cite{CO-DSB-VMP:18}, we encode the input signals as binary time series taking the value of 1 when the corresponding task-related stimulus occurs and 0 otherwise. The output signals consist of minimally pre-processed blood-oxygen-level-dependent (BOLD) time series associated with the fMRI measurements at different regions of the brain (see also \methods). In our numerical study, we parcellated the brain into $p=148$ brain regions (74 regions per hemisphere) according to the Destrieux 2009 atlas \cite{CD-BF-AD-EH:10}. Further, as a baseline for comparison, we approximate the dynamics of the functional network with a low-dimensional ($n=20$) linear model computed via the approach described in \cite{CO-DSB-VMP:18}, which has been shown to accurately capture the underlying network dynamics. 

In Fig.~\ref{fig:fMRI}(b), we plot the inputs (\emph{top}) and outputs (\emph{center}) of one subject for the first sequence of five motor tasks. The bottom plot of the same figure shows the outputs obtained by approximating the network dynamics with the above-mentioned linear model. In Fig.~\ref{fig:fMRI}(c), we compare the performance of the minimum-energy data-driven control in \eqref{eq:dd-min-en} with the model-based one, assuming that the network obeys the dynamics of the approximate linear model. We choose a control horizon $T=100$, form the data matrices in \eqref{eq:data} by sliding a window of fixed size $T$ over the available fMRI data, and consider a set of 20 orthogonal targets  \smash{$\{y_{\text{f},i}\}_{i=1}^{20}$}  corresponding to eigenvectors of the estimated $T$-step controllability Gramian (see \methods\ for further details). The top plot of Fig.~\ref{fig:fMRI}(c) reports the error (normalized by the output dimension) in the final state of the two strategies, while the bottom plot shows the corresponding control energy (that is, the norm of the control input). In the plots, the targets are ordered from the most ($y_{\text{f},1}$) to the least ($y_{\text{f},20}$) controllable. The data-driven and the model-based inputs exhibit an almost identical behavior with reference to the most controllable targets. As we shift towards the least controllable targets, the data-driven strategy yields larger errors but, at the same time, requires less energy to be implemented, thus being potentially more feasible in practice. Importantly, since the underlying brain dynamics are not known, errors in the final state are computed using the identified linear dynamical model. It is thus expected that data-driven inputs yield larger errors in the final state than model-based inputs, although these errors may not correspond to control inaccuracies when applying the data-driven inputs to the actual brain dynamics. Ultimately, our numerical study suggests that the data-driven framework could represent a viable alternative to the classic model-based approach (e.g., see \cite{SG-FP-MC-QKT-BYA-AEK-JDM-JMV-MBM-STG-DSB:15,SD-SG:20,TM-GB-DSB-FP:19b}) to infer controllability properties of brain networks, and (by suitably modulating the reconstructed inputs) enforce desired functional configurations in a non-invasive manner and without requiring real-time measurements.

\section{Discussion}
In this paper we present a framework to control complex dynamical networks from data generated by non-optimal (and possibly random) experiments. We show that optimal point-to-point controls to reach a desired target state, including the widely used minimum-energy control input, can be determined exactly from data. We provide closed-form and approximate data-based expressions of these control inputs and characterize the minimum number of samples needed to compute them. Further, we show by means of numerical simulations that data-driven inputs are more accurate and computationally more efficient than model-based ones, and can be used to analyze and manipulate the controllability properties of~real~networks.

More generally, our framework and results suggest that many network control problems may be solved by simply relying on experimental data, thus promoting a new, exciting, and practical line of research in the field of complex networks. Because of the abundance of data in modern applications and the computationally appealing properties of data-driven controls, we expect that this new line of research will benefit a broad range of research communities, spanning from engineering to biology, which employ control-theoretic methods and tools to comprehend and manipulate complex networked phenomena.

Some limitations of this study should also be acknowledged and discussed. First, in our work we consider networks governed by linear dynamics. On the one hand, this is a restrictive assumption since many real-world networks are inherently nonlinear. On the other hand, linear models are used successfully to approximate the behavior of nonlinear dynamical networks around desired operating points, and capture more explicitly the impact of the network topology. Second, in many cases a closed-loop control strategy is preferable than a point-to-point one, especially if the control objective is to stabilize an equilibrium when external disturbances corrupt the dynamics. However, we stress that point-to-point controls, in addition to being able to steer the network to arbitrary configurations, are extensively used to characterize the fundamental control properties and limitations in networks of dynamical nodes. For instance, the expressions we provide for point-to-point control can also lead to novel methods to study the energetic limitations of controlling complex networks \cite{FP-SZ-FB:13q}, select sensors and actuators for optimized estimation and control \cite{THS-FLC-JL:16}, and design optimized network structures \cite{SZ-FP:16a}. Finally, although we provide data-driven expressions that compensate for the effect of noise in the limit of infinite data, we do not provide non-asymptotic guarantees on the reconstruction error. Overcoming these limitations represents a compelling direction of future work, which can strengthen the relevance and applicability of our data-driven control framework, and ultimately lead to viable control methods for complex networks.

\section{Methods}

\subsection{Model-based expressions of optimal controls} 
The model-based solution to \eqref{eq:optimal_control} can be written in batch form as
\begin{linenomath*} 
\begin{equation}\label{eq:mb-optimal-control}
	u^{\star}_{0:T-1} = (I - K_{\mc C_{T}}(MK_{\mc C_{T}})^{\dagger}M)\mc C_{T}^{\dagger}\subscr{y}{f},
\end{equation}
\end{linenomath*}
where $\mc C_{T}=[CB \ CAB \ \cdots\ CA^{T-1}B]$ is the $T$-step output controllability matrix of the dynamical network in \eqref{eq:sys}, \smash{$K_{Y_{T}}$} denotes a basis of the kernel of $\mc C_{T}$, and $M$ is any matrix satisfying $M^{\transpose}M=\mc H_{T}^{\transpose}Q\mc H_{T} + R$,~with
\begin{linenomath*} 
\begin{align*}
	\mc H_{T} = \begin{bmatrix} 
		0 & \cdots & \cdots  & 0 & CB  \\
		\vdots & \cdots & 0 & CB & CAB  \\
		\vdots & \iddots & \iddots & \iddots & \vdots \\
		0 & CB & CAB & \cdots & CA^{T-2}B \\
		\end{bmatrix},
\end{align*}
\end{linenomath*}
and $0$ entries denoting $p\times m$ zero matrices. If $Q=0$ and $R=I$ (minimum-energy control input), \eqref{eq:mb-optimal-control} simplifies to \smash{$u^{\star}_{0:T-1} = \mc C_{T}^{\dagger}\subscr{y}{f}$}. Alternatively, if the network is target controllable, the minimum-energy input can be compactly written as
\begin{linenomath*} 
\begin{equation}\label{eq:mb-min-energy}
	u^{\star}(t) = B^{\transpose}A^{T-t-1}C^{\transpose} \mc W_{T}^{-1} \subscr{y}{f}, \quad t=0,1,2,\dots,T-1.
\end{equation}
\end{linenomath*}
where $\mc W_{T}$ denotes the $T$-step output controllability Gramian of the dynamical network in \eqref{eq:sys} 
\begin{linenomath*} 
\begin{equation}
\mc W_{T} = \mc C_{T} \mc C_{T}^{\transpose} =\sum_{t=0}^{T-1} CA^{t}BB^{\transpose}(A^{\transpose})^{t}C^{\transpose},
\end{equation}
which is invertible if and only if the network is target controllable.
\end{linenomath*}
Equation \eqref{eq:mb-min-energy} is the classic (Gramian-based) expression of the minimum-energy control input \cite{TK:80}. It is well-known that this expression is numerically unstable, even for moderate size systems, e.g.,~see~\cite{JS-AEM:13}.

\subsection{Subspace-based system identification} Given the data matrices $U_{0:T-1}$ and $Y_{T}$ as defined in \eqref{eq:data} and assuming that $C=I$, a simple deterministic subspace-based procedure \cite[Ch.~6]{TK:05} to estimate the matrices $A$ and $B$ from the available data consists of the following two steps:
\begin{enumerate}
\item Compute an estimate of the $T$-step controllability matrix of the network as the solution of the minimization problem
\begin{linenomath*} 
  \begin{equation} \label{eq:C-est}
    \hat{\mc C}_{T} = \arg\min_{\mc C_{T}} \left\|Y_{T}-{\mc C_{T}}U_{0:T-1}\right\|^{2}_{\text{F}},
  \end{equation}
where $\|\cdot\|_{\text{F}}$ denotes the Frobenius norm of a matrix. The solution to \eqref{eq:C-est} has the form
  \smash{$\hat{\mc C}_{T} = Y_{T}U_{0:T-1}^{\dagger}$}.
\end{linenomath*}
\item In view of the definition of the controllability matrix, obtain an estimate of the matrix $B$ by extracting the first $m$ columns of
  \smash{$\hat{\mc C}_{T}$}. Namely,
  \smash{$\hat B = [\hat{\mc C}_{T,}]_{:,1:m}$}, where $[X]_{:,i:j}$ indicates the sub-matrix of $X$ obtained from keeping the entries
  from the $i$-th to $j$-th columns and all of its rows. An estimate of the matrix $A$ can be obtained as the solution to~the~least-squares~problem
\begin{linenomath*} 
\begin{equation*}
	\hat A = \arg\min_{A} \left\| [\hat{\mc C}_{T}]_{:,m+1:mT} - A \,[\hat{\mc C}_{T}]_{:,1:(T-1)m}\right\|^{2}_{\text{F}},
\end{equation*}
\end{linenomath*}
which yields the matrix $\hat A=[\hat{\mc C}_{T}]_{:,m+1:mT}[\hat{\mc C}_{T}]_{:,1:(T-1)m}^{\dagger}$.
\end{enumerate}
If the data are noiseless, the system is controllable in $T-1$ steps, and $U_{0:T-1}$ has full row rank, then this procedure provably returns correct estimates~of~$A$~and~$B$ \cite{TK:05}.

\subsection{Power-grid network dynamics, parameters, and data generation} 

The short-term electromechanical behavior of generators $\{2,\dots,10\}$ of the New England power-grid network are modeled by the swing equations \cite{PK:94}:
\begin{equation}\label{eq:swing}
\begin{aligned}
\dot \delta_{i} &= \omega_{i},\\
\frac{H_{i}}{\pi f_{b}}\dot \omega_{i}  & =  -D_{i}\omega_{i} + P_{\text{m}i} -G_{ii}E_{i}^{2}+\sum_{j=1,j\ne i}^{10} E_{i}E_{j}\left(G_{ij}\cos(\delta_{i}-\delta_{j})+B_{ij}\sin(\delta_{i}-\delta_{j})\right).
\end{aligned}
\end{equation}
where $\delta_{i}$ is the angular position or phase of the rotor in generator $i$  with respect to generator $1$, and where $\omega_{i}$  is the deviation of the rotor speed or frequency in generator $i$ relative to the nominal angular frequency $2\pi f_{b}$. The generator $1$ is assumed to be connected to an infinite bus, and has constant phase and frequency. The parameters $H_{i}$ and $D_{i}$ are the  inertia constant and damping coefficient, respectively, of generator $i$. The parameter $G_{ii}$ is the internal conductance of generator $i$, and $G_{ij} + \mathrm{i}B_{ij}$ (where $\mathrm{i}$ is the imaginary unit) is the transfer impedance between generators $i$ and $j$. The parameter $P_{\text{m}i}$ denotes the mechanical input power of generator $i$ and $E_{i}$ denotes the internal voltage of generator $i$. The values of parameters $f_{b}$, $H_{i}$, $D_{i}$, $G_{ij}$, $B_{ij}$, and $P_{\text{m}i}$ in the non-faulty and faulty configuration are taken from  \cite{YS-IM-TH:11}, while the voltages $E_{i}$ and initial conditions ($\delta_{i} (0)$, $\omega_{i} (0) = 0$) are fixed using a power flow computation. In our numerical study, we discretize the dynamics \eqref{eq:swing} using a forward Euler method with sampling time $T_{s}=2.5\times 10^{-4}$s. Data are generated by applying a Gaussian i.i.d.~perturbation with zero mean and variance $0.01$ to each frequency $\omega_{i}$. The initial condition of each experiment is computed by adding a Gaussian i.i.d.~perturbation with zero mean and variance $0.01$ to the steady-state values of $\delta_{i}$ and  $\omega_{i}$ of the swing dynamics \eqref{eq:swing}. 

\subsection{Task-fMRI dataset, pre-processing pipeline, and identification setup} The motor task fMRI data used in our numerical study are extracted from the HCP S1200 release \cite{DCVE-et-al:13,HCP}. The details for data acquisition and experiment design can be found in \cite{HCP}. The BOLD measurements have been pre-processed according to the minimal pipeline described in \cite{MFG-et-al:13}, and, as in \cite{CO-DSB-VMP:18}, filtered with a band-pass filter to attenuate the frequencies outside the 0.06--0.12 Hz band. Further, as common practice, the effect of the physiological signals (cardiac, respiratory, and head motion signals) is removed from the BOLD measurements by means of the regression procedure in \cite{CO-DSB-VMP:18}. The data matrices in \eqref{eq:data} are generated via a sliding window of fixed length $T=100$ with initial time in the interval $[-90,10]$. We assume that the inputs and states are zero for times less than or equal to 10, i.e., the instant at which the first task condition is issued. We approximate the input-output dynamics with a linear model with state dimension $n=20$ computed using input-output data in the interval $[0,150]$ and the identification procedure~detailed~in~\cite{CO-DSB-VMP:18}.  When the estimated network matrix $A$ has unstable eigenvalues, we stabilize $A$ by diving it by $\rho(A)+0.01$, where $\rho(A)$ denotes the spectral radius of $A$. Other identification~parameters~are~as~in~\cite{CO-DSB-VMP:18}.

\subsection{Computational details} 
All numerical simulations have been performed via standard linear-algebra LAPACK routines available as built-in functions in
Matlab\textsuperscript{\textregistered} R2019b, running on a 2.6 GHz Intel Core i5 processor with 8 GB of RAM. In particular, for the
computation of pseudoinverses we use the singular value decomposition method (command \texttt{pinv} in Matlab\textsuperscript{\textregistered}) with a threshold~of~$10^{-8}$.

\subsection{Materials and data availability} 
Data were provided (in part) by the Human Connectome Project, WU-Minn Consortium (Principal Investigators: David Van Essen and Kamil Ugurbil; 1U54MH091657) funded by the 16 NIH Institutes and Centers that support the NIH Blueprint for Neuroscience Research; and by the McDonnell Center for Systems Neuroscience at Washington University. The code and data used in this study are freely available in the public GitHub repository: \href{https://github.com/baggiogi/data_driven_control}{\texttt{https://github.com/baggiogi/data\_driven\_control}}.

\section*{Acknowledgements}
This work was supported in part by awards AFOSR-FA9550-19-1-0235 and ARO-71603NSYIP, and by MIUR (Italian Minister for Education)  under the initiative ``Departments of Excellence'' (Law 232/2016).

\section*{Author contributions}
G.B, D.S.B, and F.P. contributed to the conceptual and theoretical aspects of the study, wrote the manuscript and the \SI. G.B. carried out the numerical studies and prepared the figures.

\section*{Competing interests}
The authors declare no competing interests.

\bibliography{alias,FP,Main,New}

\begin{thebibliography}{10}
\expandafter\ifx\csname url\endcsname\relax
  \def\url#1{\texttt{#1}}\fi
\expandafter\ifx\csname urlprefix\endcsname\relax\def\urlprefix{URL }\fi
\providecommand{\bibinfo}[2]{#2}
\providecommand{\eprint}[2][]{\url{#2}}

\bibitem{SL-PP-AK-JI-DQ:18}
\bibinfo{author}{Levine, S.}, \bibinfo{author}{Pastor, P.},
  \bibinfo{author}{Krizhevsky, A.}, \bibinfo{author}{Ibarz, J.} \&
  \bibinfo{author}{Quillen, D.}
\newblock \bibinfo{title}{Learning hand-eye coordination for robotic grasping
  with deep learning and large-scale data collection}.
\newblock \emph{\bibinfo{journal}{The International Journal of Robotics
  Research}} \textbf{\bibinfo{volume}{37}}, \bibinfo{pages}{421--436}
  (\bibinfo{year}{2018}).

\bibitem{VM:13}
\bibinfo{author}{Marx, V.}
\newblock \bibinfo{title}{Biology: {T}he big challenges of big data}.
\newblock \emph{\bibinfo{journal}{Nature}} \textbf{\bibinfo{volume}{498}},
  \bibinfo{pages}{255--260} (\bibinfo{year}{2013}).

\bibitem{TJS-PSC-JAM:14}
\bibinfo{author}{Sejnowski, T.~J.}, \bibinfo{author}{Churchland, P.~S.} \&
  \bibinfo{author}{Movshon, J.~A.}
\newblock \bibinfo{title}{Putting big data to good use in neuroscience}.
\newblock \emph{\bibinfo{journal}{Nature neuroscience}}
  \textbf{\bibinfo{volume}{17}}, \bibinfo{pages}{1440} (\bibinfo{year}{2014}).

\bibitem{LE-JL:14}
\bibinfo{author}{Einav, L.} \& \bibinfo{author}{Levin, J.}
\newblock \bibinfo{title}{Economics in the age of big data}.
\newblock \emph{\bibinfo{journal}{Science}} \textbf{\bibinfo{volume}{346}},
  \bibinfo{pages}{1243089} (\bibinfo{year}{2014}).

\bibitem{NBT:13}
\bibinfo{author}{Turk-Browne, N.~B.}
\newblock \bibinfo{title}{Functional interactions as big data in the human
  brain}.
\newblock \emph{\bibinfo{journal}{Science}} \textbf{\bibinfo{volume}{342}},
  \bibinfo{pages}{580--584} (\bibinfo{year}{2013}).

\bibitem{AB:10}
\bibinfo{author}{Bose, A.}
\newblock \bibinfo{title}{Smart transmission grid applications and their
  supporting infrastructure}.
\newblock \emph{\bibinfo{journal}{IEEE Transactions on Smart Grid}}
  \textbf{\bibinfo{volume}{1}}, \bibinfo{pages}{11--19} (\bibinfo{year}{2010}).

\bibitem{YL-YD-WK-ZL-FYW:14}
\bibinfo{author}{Lv, Y.}, \bibinfo{author}{Duan, Y.}, \bibinfo{author}{Kang,
  W.}, \bibinfo{author}{Li, Z.} \& \bibinfo{author}{Wang, F.-Y.}
\newblock \bibinfo{title}{Traffic flow prediction with big data: a deep
  learning approach}.
\newblock \emph{\bibinfo{journal}{IEEE Transactions on Intelligent
  Transportation Systems}} \textbf{\bibinfo{volume}{16}},
  \bibinfo{pages}{865--873} (\bibinfo{year}{2014}).

\bibitem{YYL-JJS-ALB:11}
\bibinfo{author}{Liu, Y.~Y.}, \bibinfo{author}{Slotine, J.~J.} \&
  \bibinfo{author}{Barab{\'a}si, A.~L.}
\newblock \bibinfo{title}{Controllability of complex networks}.
\newblock \emph{\bibinfo{journal}{Nature}} \textbf{\bibinfo{volume}{473}},
  \bibinfo{pages}{167--173} (\bibinfo{year}{2011}).

\bibitem{FP-SZ-FB:13q}
\bibinfo{author}{Pasqualetti, F.}, \bibinfo{author}{Zampieri, S.} \&
  \bibinfo{author}{Bullo, F.}
\newblock \bibinfo{title}{Controllability metrics, limitations and algorithms
  for complex networks}.
\newblock \emph{\bibinfo{journal}{IEEE Transactions on Control of Network
  Systems}} \textbf{\bibinfo{volume}{1}}, \bibinfo{pages}{40--52}
  (\bibinfo{year}{2014}).

\bibitem{NB-GB-SZ:17}
\bibinfo{author}{Bof, N.}, \bibinfo{author}{Baggio, G.} \&
  \bibinfo{author}{Zampieri, S.}
\newblock \bibinfo{title}{On the role of network centrality in the
  controllability of complex networks}.
\newblock \emph{\bibinfo{journal}{IEEE Transactions on Control of Network
  Systems}} \textbf{\bibinfo{volume}{4}}, \bibinfo{pages}{643--653}
  (\bibinfo{year}{2017}).

\bibitem{GY-GT-BB-JS-YL-AB:15}
\bibinfo{author}{Yan, G.} \emph{et~al.}
\newblock \bibinfo{title}{Spectrum of controlling and observing complex
  networks}.
\newblock \emph{\bibinfo{journal}{Nature Physics}}
  \textbf{\bibinfo{volume}{11}}, \bibinfo{pages}{779--786}
  (\bibinfo{year}{2015}).

\bibitem{SG-FP-MC-QKT-BYA-AEK-JDM-JMV-MBM-STG-DSB:15}
\bibinfo{author}{Gu, S.} \emph{et~al.}
\newblock \bibinfo{title}{Controllability of structural brain networks}.
\newblock \emph{\bibinfo{journal}{Nature Communications}}
  \textbf{\bibinfo{volume}{6}} (\bibinfo{year}{2015}).

\bibitem{YYL-ALB:16}
\bibinfo{author}{Liu, Y.-Y.} \& \bibinfo{author}{Barab\'asi, A.-L.}
\newblock \bibinfo{title}{Control principles of complex systems}.
\newblock \emph{\bibinfo{journal}{Reviews in {M}odern {P}hysics}}
  \textbf{\bibinfo{volume}{88}}, \bibinfo{pages}{035006}
  (\bibinfo{year}{2016}).

\bibitem{GL-CA:16}
\bibinfo{author}{Lindmark, G.} \& \bibinfo{author}{Altafini, C.}
\newblock \bibinfo{title}{Minimum energy control for complex networks}.
\newblock \emph{\bibinfo{journal}{{S}cientific {R}eports}}
  \textbf{\bibinfo{volume}{8}}, \bibinfo{pages}{3188--3202}
  (\bibinfo{year}{2018}).

\bibitem{JC-SW:08}
\bibinfo{author}{Gon{\c{c}}alves, J.} \& \bibinfo{author}{Warnick, S.}
\newblock \bibinfo{title}{Necessary and sufficient conditions for dynamical
  structure reconstruction of lti networks}.
\newblock \emph{\bibinfo{journal}{IEEE Transactions on Automatic Control}}
  \textbf{\bibinfo{volume}{53}}, \bibinfo{pages}{1670--1674}
  (\bibinfo{year}{2008}).

\bibitem{SGS-MT:11}
\bibinfo{author}{Shandilya, S.~G.} \& \bibinfo{author}{Timme, M.}
\newblock \bibinfo{title}{Inferring network topology from complex dynamics}.
\newblock \emph{\bibinfo{journal}{New Journal of Physics}}
  \textbf{\bibinfo{volume}{13}}, \bibinfo{pages}{013004}
  (\bibinfo{year}{2011}).

\bibitem{MTA-JAM-GL-ALB-YYL:17}
\bibinfo{author}{Angulo, M.~T.}, \bibinfo{author}{Moreno, J.~A.},
  \bibinfo{author}{Lippner, G.}, \bibinfo{author}{Barab{\'a}si, A.-L.} \&
  \bibinfo{author}{Liu, Y.-Y.}
\newblock \bibinfo{title}{Fundamental limitations of network reconstruction
  from temporal data}.
\newblock \emph{\bibinfo{journal}{Journal of the Royal Society Interface}}
  \textbf{\bibinfo{volume}{14}}, \bibinfo{pages}{20160966}
  (\bibinfo{year}{2017}).

\bibitem{DA-AC-DK-CM:09}
\bibinfo{author}{Achlioptas, D.}, \bibinfo{author}{Clauset, A.},
  \bibinfo{author}{Kempe, D.} \& \bibinfo{author}{Moore, C.}
\newblock \bibinfo{title}{On the bias of traceroute sampling: or, power-law
  degree distributions in regular graphs}.
\newblock \emph{\bibinfo{journal}{Journal of the ACM (JACM)}}
  \textbf{\bibinfo{volume}{56}}, \bibinfo{pages}{1--28} (\bibinfo{year}{2009}).

\bibitem{MSH-KJG:10}
\bibinfo{author}{Handcock, M.~S.} \& \bibinfo{author}{Gile, K.~J.}
\newblock \bibinfo{title}{Modeling social networks from sampled data}.
\newblock \emph{\bibinfo{journal}{The Annals of Applied Statistics}}
  \textbf{\bibinfo{volume}{4}}, \bibinfo{pages}{5} (\bibinfo{year}{2010}).

\bibitem{JS-AEM:13}
\bibinfo{author}{Sun, J.} \& \bibinfo{author}{Motter, A.~E.}
\newblock \bibinfo{title}{Controllability transition and nonlocality in network
  control}.
\newblock \emph{\bibinfo{journal}{Physical Review Letters}}
  \textbf{\bibinfo{volume}{110}}, \bibinfo{pages}{208701}
  (\bibinfo{year}{2013}).

\bibitem{LZW-YZC-WXW-YCL:17}
\bibinfo{author}{Wang, L.-Z.}, \bibinfo{author}{Chen, Y.-Z.},
  \bibinfo{author}{Wang, W.-X.} \& \bibinfo{author}{Lai, Y.-C.}
\newblock \bibinfo{title}{Physical controllability of complex networks}.
\newblock \emph{\bibinfo{journal}{Scientific reports}}
  \textbf{\bibinfo{volume}{7}}, \bibinfo{pages}{40198} (\bibinfo{year}{2017}).

\bibitem{SL-CF-TD-PA:16}
\bibinfo{author}{Levine, S.}, \bibinfo{author}{Finn, C.},
  \bibinfo{author}{Darrell, T.} \& \bibinfo{author}{Abbeel, P.}
\newblock \bibinfo{title}{End-to-end training of deep visuomotor policies}.
\newblock \emph{\bibinfo{journal}{The Journal of Machine Learning Research}}
  \textbf{\bibinfo{volume}{17}}, \bibinfo{pages}{1334--1373}
  (\bibinfo{year}{2016}).

\bibitem{DS-et-al:17}
\bibinfo{author}{Silver, D.} \emph{et~al.}
\newblock \bibinfo{title}{Mastering the game of {G}o without human knowledge}.
\newblock \emph{\bibinfo{journal}{Nature}} \textbf{\bibinfo{volume}{550}},
  \bibinfo{pages}{354} (\bibinfo{year}{2017}).

\bibitem{MG:05}
\bibinfo{author}{Gevers, M.}
\newblock \bibinfo{title}{Identification for control: From the early
  achievements to the revival of experiment design}.
\newblock \emph{\bibinfo{journal}{European Journal of Control}}
  \textbf{\bibinfo{volume}{11}}, \bibinfo{pages}{1--18} (\bibinfo{year}{2005}).

\bibitem{SLB-JNK:19}
\bibinfo{author}{Brunton, S.~L.} \& \bibinfo{author}{Kutz, J.~N.}
\newblock \emph{\bibinfo{title}{Data-driven science and engineering: {M}achine
  learning, dynamical systems, and control}} (\bibinfo{publisher}{Cambridge
  University Press}, \bibinfo{year}{2019}).

\bibitem{FLL-DV-KGV:12}
\bibinfo{author}{Lewis, F.~L.}, \bibinfo{author}{Vrabie, D.} \&
  \bibinfo{author}{Vamvoudakis, K.~G.}
\newblock \bibinfo{title}{Reinforcement learning and feedback control: Using
  natural decision methods to design optimal adaptive controllers}.
\newblock \emph{\bibinfo{journal}{{IEEE} Control Systems Magazine}}
  \textbf{\bibinfo{volume}{32}}, \bibinfo{pages}{76--105}
  (\bibinfo{year}{2012}).

\bibitem{BR:18}
\bibinfo{author}{Recht, B.}
\newblock \bibinfo{title}{A tour of reinforcement learning: The view from
  continuous control}.
\newblock \emph{\bibinfo{journal}{Annual Review of Control, Robotics, and
  Autonomous Systems}}  (\bibinfo{year}{2018}).

\bibitem{DAB-MT-AGA:06}
\bibinfo{author}{Bristow, D.~A.}, \bibinfo{author}{Tharayil, M.} \&
  \bibinfo{author}{Alleyne, A.~G.}
\newblock \bibinfo{title}{A survey of iterative learning control}.
\newblock \emph{\bibinfo{journal}{IEEE control systems magazine}}
  \textbf{\bibinfo{volume}{26}}, \bibinfo{pages}{96--114}
  (\bibinfo{year}{2006}).

\bibitem{KJA-BW:73}
\bibinfo{author}{{\AA}str{\"o}m, K.~J.} \& \bibinfo{author}{Wittenmark, B.}
\newblock \bibinfo{title}{On self tuning regulators}.
\newblock \emph{\bibinfo{journal}{Automatica}} \textbf{\bibinfo{volume}{9}},
  \bibinfo{pages}{185--199} (\bibinfo{year}{1973}).

\bibitem{IM-PR:08}
\bibinfo{author}{Markovsky, I.} \& \bibinfo{author}{Rapisarda, P.}
\newblock \bibinfo{title}{Data-driven simulation and control}.
\newblock \emph{\bibinfo{journal}{International Journal of Control}}
  \textbf{\bibinfo{volume}{81}}, \bibinfo{pages}{1946--1959}
  (\bibinfo{year}{2008}).

\bibitem{CDP-PT:19}
\bibinfo{author}{Persis, C.~D.} \& \bibinfo{author}{Tesi, P.}
\newblock \bibinfo{title}{Formulas for data-driven control: Stabilization,
  optimality and robustness}.
\newblock \emph{\bibinfo{journal}{IEEE Transactions on Automatic Control}}
  \textbf{\bibinfo{volume}{65}}, \bibinfo{pages}{909--924}
  (\bibinfo{year}{2020}).

\bibitem{DPB-JNT:96}
\bibinfo{author}{Bertsekas, D.~P.} \& \bibinfo{author}{Tsitsiklis, J.~N.}
\newblock \emph{\bibinfo{title}{Neuro-dynamic programming}},
  vol.~\bibinfo{volume}{5} (\bibinfo{publisher}{Athena Scientific Belmont, MA},
  \bibinfo{year}{1996}).

\bibitem{KJA-TH:95}
\bibinfo{author}{{\AA}str{\"o}m, K.~J.} \& \bibinfo{author}{H{\"a}gglund, T.}
\newblock \emph{\bibinfo{title}{{PID} controllers: theory, design, and
  tuning}}, vol.~\bibinfo{volume}{2} (\bibinfo{publisher}{Instrument society of
  America Research Triangle Park, NC}, \bibinfo{year}{1995}).

\bibitem{JG-YYL-RMD-ALB:14}
\bibinfo{author}{Gao, J.}, \bibinfo{author}{Liu, Y.-Y.},
  \bibinfo{author}{D'Souza, R.~M.} \& \bibinfo{author}{Barab{\'a}si, A.~L.}
\newblock \bibinfo{title}{Target control of complex networks}.
\newblock \emph{\bibinfo{journal}{Nature communications}}
  \textbf{\bibinfo{volume}{5}}, \bibinfo{pages}{5415} (\bibinfo{year}{2014}).

\bibitem{IK-AS-FS:17b}
\bibinfo{author}{Klickstein, I.}, \bibinfo{author}{Shirin, A.} \&
  \bibinfo{author}{Sorrentino, F.}
\newblock \bibinfo{title}{Energy scaling of targeted optimal control of complex
  networks}.
\newblock \emph{\bibinfo{journal}{Nature communications}}
  \textbf{\bibinfo{volume}{8}}, \bibinfo{pages}{15145} (\bibinfo{year}{2017}).

\bibitem{TK:80}
\bibinfo{author}{Kailath, T.}
\newblock \emph{\bibinfo{title}{Linear Systems}}
  (\bibinfo{publisher}{Prentice-Hall}, \bibinfo{year}{1980}).

\bibitem{SD-HM-NM-BR-ST:18}
\bibinfo{author}{Dean, S.}, \bibinfo{author}{Mania, H.},
  \bibinfo{author}{Matni, N.}, \bibinfo{author}{Recht, B.} \&
  \bibinfo{author}{Tu, S.}
\newblock \bibinfo{title}{On the sample complexity of the linear quadratic
  regulator}.
\newblock \emph{\bibinfo{journal}{Foundations of Computational Mathematics}}
  \bibinfo{pages}{1--47} (\bibinfo{year}{2019}).

\bibitem{AB-TNEG:03}
\bibinfo{author}{Ben-Israel, A.} \& \bibinfo{author}{Greville, T.~N.~E.}
\newblock \emph{\bibinfo{title}{Generalized inverses: theory and
  applications}}, vol.~\bibinfo{volume}{15} of \emph{\bibinfo{series}{CMS Books
  in Mathematics}} (\bibinfo{publisher}{Springer-Verlag New York},
  \bibinfo{year}{2003}), \bibinfo{edition}{2nd} edn.

\bibitem{PPE-VC-SW:13}
\bibinfo{author}{Par{\'e}, P.~E.}, \bibinfo{author}{Chetty, V.} \&
  \bibinfo{author}{Warnick, S.}
\newblock \bibinfo{title}{On the necessity of full-state measurement for
  state-space network reconstruction}.
\newblock In \emph{\bibinfo{booktitle}{2013 IEEE Global Conference on Signal
  and Information Processing}}, \bibinfo{pages}{803--806}
  (\bibinfo{organization}{IEEE}, \bibinfo{year}{2013}).

\bibitem{YS-IM-TH:11}
\bibinfo{author}{Susuki, Y.}, \bibinfo{author}{Mezi{\'c}, I.} \&
  \bibinfo{author}{Hikihara, T.}
\newblock \bibinfo{title}{Coherent swing instability of power grids}.
\newblock \emph{\bibinfo{journal}{Journal of nonlinear science}}
  \textbf{\bibinfo{volume}{21}}, \bibinfo{pages}{403--439}
  (\bibinfo{year}{2011}).

\bibitem{SPC-LWK-AEM:13}
\bibinfo{author}{Cornelius, S.~P.}, \bibinfo{author}{Kath, W.~L.} \&
  \bibinfo{author}{Motter, A.~E.}
\newblock \bibinfo{title}{Realistic control of network dynamics}.
\newblock \emph{\bibinfo{journal}{Nature Communications}}
  \textbf{\bibinfo{volume}{4}} (\bibinfo{year}{2013}).

\bibitem{JWSP-FD-FB:16}
\bibinfo{author}{Simpson-Porco, J.~W.}, \bibinfo{author}{D{\"o}rfler, F.} \&
  \bibinfo{author}{Bullo, F.}
\newblock \bibinfo{title}{Voltage collapse in complex power grids}.
\newblock \emph{\bibinfo{journal}{Nature Communications}}
  \textbf{\bibinfo{volume}{7}}, \bibinfo{pages}{1--8} (\bibinfo{year}{2016}).

\bibitem{DCVE-et-al:13}
\bibinfo{author}{Van~Essen, D.~C.} \emph{et~al.}
\newblock \bibinfo{title}{The {WU}-{Minn} human connectome project: an
  overview}.
\newblock \emph{\bibinfo{journal}{Neuroimage}} \textbf{\bibinfo{volume}{80}},
  \bibinfo{pages}{62--79} (\bibinfo{year}{2013}).

\bibitem{CO-DSB-VMP:18}
\bibinfo{author}{Becker, C.~O.}, \bibinfo{author}{Bassett, D.~S.} \&
  \bibinfo{author}{Preciado, V.~M.}
\newblock \bibinfo{title}{Large-scale dynamic modeling of task-f{MRI} signals
  via subspace system identification}.
\newblock \emph{\bibinfo{journal}{Journal of neural engineering}}
  \textbf{\bibinfo{volume}{15}}, \bibinfo{pages}{066016}
  (\bibinfo{year}{2018}).

\bibitem{CD-BF-AD-EH:10}
\bibinfo{author}{Destrieux, C.}, \bibinfo{author}{Fischl, B.},
  \bibinfo{author}{Dale, A.} \& \bibinfo{author}{Halgren, E.}
\newblock \bibinfo{title}{Automatic parcellation of human cortical gyri and
  sulci using standard anatomical nomenclature}.
\newblock \emph{\bibinfo{journal}{Neuroimage}} \textbf{\bibinfo{volume}{53}},
  \bibinfo{pages}{1--15} (\bibinfo{year}{2010}).

\bibitem{SD-SG:20}
\bibinfo{author}{Deng, S.} \& \bibinfo{author}{Gu, S.}
\newblock \bibinfo{title}{Controllability analysis of functional brain
  networks}.
\newblock \emph{\bibinfo{journal}{arXiv preprint arXiv:2003.08278}}
  (\bibinfo{year}{2020}).

\bibitem{TM-GB-DSB-FP:19b}
\bibinfo{author}{Menara, T.}, \bibinfo{author}{Baggio, G.},
  \bibinfo{author}{Bassett, D.~S.} \& \bibinfo{author}{Pasqualetti, F.}
\newblock \bibinfo{title}{A framework to control functional connectivity in the
  human brain}.
\newblock In \emph{\bibinfo{booktitle}{{IEEE} Conf.\ on Decision and Control}},
  \bibinfo{pages}{4697--4704} (\bibinfo{address}{Nice, France},
  \bibinfo{year}{2019}).

\bibitem{THS-FLC-JL:16}
\bibinfo{author}{Summers, T.~H.}, \bibinfo{author}{Cortesi, F.~L.} \&
  \bibinfo{author}{Lygeros, J.}
\newblock \bibinfo{title}{On submodularity and controllability in complex
  dynamical networks}.
\newblock \emph{\bibinfo{journal}{IEEE Transactions on Control of Network
  Systems}} \textbf{\bibinfo{volume}{3}}, \bibinfo{pages}{91--101}
  (\bibinfo{year}{2016}).

\bibitem{SZ-FP:16a}
\bibinfo{author}{Zhao, S.} \& \bibinfo{author}{Pasqualetti, F.}
\newblock \bibinfo{title}{Networks with diagonal controllability gramians:
  Analysis, graphical conditions, and design algorithms}.
\newblock \emph{\bibinfo{journal}{Automatica}} \textbf{\bibinfo{volume}{102}},
  \bibinfo{pages}{10--18} (\bibinfo{year}{2019}).

\bibitem{TK:05}
\bibinfo{author}{Katayama, T.}
\newblock \emph{\bibinfo{title}{Subspace methods for system identification}}.
\newblock Communications and Control Engineering
  (\bibinfo{publisher}{Springer-Verlag London}, \bibinfo{year}{2005}).

\bibitem{PK:94}
\bibinfo{author}{Kundur, P.}
\newblock \emph{\bibinfo{title}{Power System Stability and Control}}
  (\bibinfo{publisher}{McGraw-Hill}, \bibinfo{year}{1994}).

\bibitem{HCP}
\bibinfo{title}{{WU}-{Minn}, {HCP} 1200 subjects data release reference
  manual}.
\newblock \bibinfo{howpublished}{\url{https://www. humanconnectome. org}}
  (\bibinfo{year}{2017}).
\newblock \bibinfo{note}{Accessed: 2020-03-15}.

\bibitem{MFG-et-al:13}
\bibinfo{author}{Glasser, M.~F.} \emph{et~al.}
\newblock \bibinfo{title}{The minimal preprocessing pipelines for the human
  connectome project}.
\newblock \emph{\bibinfo{journal}{Neuroimage}} \textbf{\bibinfo{volume}{80}},
  \bibinfo{pages}{105--124} (\bibinfo{year}{2013}).

\end{thebibliography}


\begin{thebibliography}{1}
\expandafter\ifx\csname url\endcsname\relax
  \def\url#1{\texttt{#1}}\fi
\expandafter\ifx\csname urlprefix\endcsname\relax\def\urlprefix{URL }\fi
\providecommand{\bibinfo}[2]{#2}
\providecommand{\eprint}[2][]{\url{#2}}

\bibitem{TK:80}
\bibinfo{author}{Kailath, T.}
\newblock \emph{\bibinfo{title}{Linear Systems}}
  (\bibinfo{publisher}{Prentice-Hall}, \bibinfo{year}{1980}).

\bibitem{AB-TNEG:03}
\bibinfo{author}{Ben-Israel, A.} \& \bibinfo{author}{Greville, T.~N.~E.}
\newblock \emph{\bibinfo{title}{Generalized inverses: theory and
  applications}}, vol.~\bibinfo{volume}{15} of \emph{\bibinfo{series}{CMS Books
  in Mathematics}} (\bibinfo{publisher}{Springer-Verlag New York},
  \bibinfo{year}{2003}), \bibinfo{edition}{2nd} edn.

\bibitem{RV:12}
\bibinfo{author}{Vershynin, R.}
\newblock \emph{\bibinfo{title}{Introduction to the non-asymptotic analysis of
  random matrices}}, \bibinfo{pages}{210--268} (\bibinfo{publisher}{Cambridge
  University Press}, \bibinfo{year}{2012}).

\bibitem{SSR:17}
\bibinfo{author}{Searle, S.~R.}
\newblock \emph{\bibinfo{title}{Matrix algebra useful for statistics}}
  (\bibinfo{publisher}{John Wiley \& Sons}, \bibinfo{year}{1982}).

\bibitem{FP-SZ-FB:13q}
\bibinfo{author}{Pasqualetti, F.}, \bibinfo{author}{Zampieri, S.} \&
  \bibinfo{author}{Bullo, F.}
\newblock \bibinfo{title}{Controllability metrics, limitations and algorithms
  for complex networks}.
\newblock \emph{\bibinfo{journal}{IEEE Transactions on Control of Network
  Systems}} \textbf{\bibinfo{volume}{1}}, \bibinfo{pages}{40--52}
  (\bibinfo{year}{2014}).

\bibitem{SD-HM-NM-BR-ST:18}
\bibinfo{author}{Dean, S.}, \bibinfo{author}{Mania, H.},
  \bibinfo{author}{Matni, N.}, \bibinfo{author}{Recht, B.} \&
  \bibinfo{author}{Tu, S.}
\newblock \bibinfo{title}{On the sample complexity of the linear quadratic
  regulator}.
\newblock \emph{\bibinfo{journal}{Foundations of Computational Mathematics}}
  \bibinfo{pages}{1--47} (\bibinfo{year}{2019}).

\bibitem{AWV00}
\bibinfo{author}{Van~der Vaart, A.~W.}
\newblock \emph{\bibinfo{title}{Asymptotic statistics}},
  vol.~\bibinfo{volume}{3} of \emph{\bibinfo{series}{Cambridge Series in
  Statistical and Probabilistic Mathematics}} (\bibinfo{publisher}{Cambridge
  University Press}, \bibinfo{year}{2000}).

\end{thebibliography}

 \nolinenumbers
 
\printfigures

\end{document}


\bibliographystyle{naturemag}
\title{Supplement to ``Data-Driven Control of Complex Networks''}
\author{Giacomo Baggio}
\affiliation{Department of Information Engineering, University of Padova, Padova, Italy}
\author{Danielle S.~Bassett}
\affiliation{Departments of Bioengineering, Physics \& Astronomy, Electrical
  \& Systems Engineering, 
  Neurology, and Psychiatry, University of Pennsylvania, Philadelphia, USA\\ Santa Fe Institute, Santa Fe, USA}
\author{Fabio Pasqualetti}
\affiliation{Department of Mechanical Engineering, University of
  California at Riverside, Riverside, USA\\ To whom correspondence should be addressed: \href{mailto:fabiopas@engr.ucr.edu}{fabiopas@engr.ucr.edu}}

\date{\today}

\maketitle

\tableofcontents

\newpage


\section{Expression of optimal data-driven controls for arbitrary $Q\succeq 0$ and $R\succ 0$}\label{sec:optimal-control-derivation}
Consider the data matrices $U_{0:T-1}$, $Y_{1:T-1}$, $Y_{T}$ as defined in Eq.~(3) of the main text. By linearity of the system, a linear combination of the input experiments (columns of $U_{0:T-1}$) yields an output that is a linear combination (with the same coefficients) of the corresponding output data; that is, for any vector $\alpha\in\mathbb{R}^{N}$, the input $u_{0:T}=U_{0:T}\alpha$ generates the outputs $y_{1:t-1}=Y_{1:T-1}\alpha$ and $y_{T}=Y_{T}\alpha$. Assume that there exists a vector $\alpha^{\star}$ such that 
\begin{linenomath*}
\begin{equation} \label{eq:alpha_star}
u^{\star}_{0:T-1}=U_{0:T-1}\alpha^{\star}
\end{equation}
\end{linenomath*}
is the optimal (w.r.t.~the cost function in Eq.~(2) of the main text) control reaching $\subscr{y}{f}$ in $T$ steps. Then, $\alpha^{\star}$ satisfies
\begin{linenomath*}
\begin{equation}\label{eq:optimal_control} \begin{aligned}
	\alpha^{\star} = \arg\min_{\alpha}\ \ &  \left\|L\alpha\right\|^{2}_{2} \\
	&\quad  \text{s.t. }  \subscr{y}{f}=Y_T \alpha .
\end{aligned}
\end{equation}
\end{linenomath*}
where $L$ is any matrix satisfying $L^{\top}L=Y_{1:T-1}^{\transpose}\, Q\, Y_{1:T-1}+U_{0:T-1}^{\transpose}\, R\, U_{0:T-1}$\footnote{Notice that such a matrix $L$ always exists since $Q\succeq 0$ and $R\succ 0$ imply that $Y_{1:T-1}^{\transpose}\, Q\, Y_{1:T-1}+U_{0:T-1}^{\transpose}\, R\, U_{0:T-1}\succeq 0$.} and $\|\cdot\|_{2}^{2}$ denotes the $\ell^{2}$-norm of a vector.
From the constraint $\subscr{y}{f}=Y_T\alpha$, the optimal $\alpha^{*}$ has the form $\alpha^{\star} = Y_{T}^{\dag}\subscr{y}{f}+K_{Y_{T}}w^{\star}$, with $w^{\star}$ being any vector such that
\begin{linenomath*}
\begin{equation} \begin{aligned}
	w^{\star} &= \arg\min_{w}\ \   \left\|L(Y_{T}^{\dag}\subscr{y}{f}+K_{Y_{T}}w)\right\|_{2}^{2},
\end{aligned}
\end{equation}
\end{linenomath*}
where $K_{Y_{T}}$ denotes a matrix whose columns form a basis of the kernel of $Y_{T}$. 
The solutions to the latter problem are of the form $w^{\star}=-(LK_{Y_{T}})^{\dag} LY_{T}^{\dag}\subscr{y}{f} + K_{LY_{T}}v$, where $K_{LY_{T}}$ is a matrix whose columns form a basis of $\Ker(LY_{T})$ and $v$ is an arbitrary vector. Since $\Ker(LY_{T})\subseteq\Ker(Y_{T})$, the optimal vector $\alpha^{\star}$ is unique and of the form $\alpha^{\star} = (I-K_{Y_{T}}(LK_{Y_{T}})^{\dag} L)Y_{T}^{\dag}\subscr{y}{f}$. This implies that the (unique) optimal control $u^{\star}_{0:T-1}$ to reach $\subscr{y}{f}$ in $T$ steps (that is, \eqref{eq:alpha_star}) can be~written~as 
\begin{linenomath*}
\begin{equation} \begin{aligned}\label{eq:dd}
	u_{0:T-1}^{\star} = U_{0:T-1}\alpha^{\star} = U_{0:T-1}(I-K_{Y_{T}}(LK_{Y_{T}})^{\dag} L)Y_{T}^{\dag}\subscr{y}{f}.
\end{aligned}
\end{equation}
\end{linenomath*}

\section{Minimum number of data to reconstruct $u_{0:T-1}^{\star}$} If the columns of $U_{0:T-1}$ span the space of all possible input sequences $\mathbb{R}^{mT}$, that is  $U_{0:T-1}$ is full-row rank, then there always exists a vector $\alpha^{\star}$ satisfying $u^{\star}_{0:T-1}=U_{0:T-1}\alpha^{\star}$, for any $\subscr{y}{f}$. Hence, $mT$ linearly independent data are sufficient to reconstruct the optimal control via \eqref{eq:dd}. In contrast, since any linear combination of optimal control inputs is still an optimal control input (for an output $\subscr{y}{f}$ which is the linear combination of the ones reached by the available optimal inputs), if $U_{0:T-1}$ contains $p$ linearly independent and optimal data, then \eqref{eq:dd} yields the optimal control input, for any target $\subscr{y}{f}$. Thus, at least $p$ linearly independent data must be collected to reconstruct the optimal control input via \eqref{eq:dd}. Finally, we observe that if $p$ linearly independent (and possibly non-optimal) data are available, then $Y_{T}$ has full row rank and the optimal vector $\alpha^{\star}$ in \eqref{eq:optimal_control} satisfies the constraint $\subscr{y}{f}=Y_T\alpha^{\star}$, so that the resulting control in \eqref{eq:dd}, although in general suboptimal, still correctly steers the output to the desired target $\subscr{y}{f}$, for any choice of $\subscr{y}{f}$.

\section{Extension to general experimental settings} When the initial state of the network is different from $x(0)=0$ across the experiments, then the optimal control input can be still reconstructed from data, provided that the initial state of each experiment can be measured. In this case, data may consist of a single uninterrupted input/output trajectory of the system, where different control experiments correspond to different segments of length $T$ of this trajectory. 

Let \smash{$X_{0}=\begin{bmatrix} x_{0}^{(1)}& x_{0}^{(2)} & \cdots & x_{0}^{(N)}\end{bmatrix}\in\mathbb{R}^{n\times N}$} denote the matrix whose columns consist of the initial states of all control experiments, and $K_{X_{0}}$ denote a basis of $\Ker(X_{0})$. Note that an input sequence $u_{0:T-1}$ expressed as a linear combination of the columns of $U_{0:T-1}K_{X_{0}}$ yields an output trajectory $y_{1:T}$ which is a linear combination (with the same coefficients) of the columns of $Y_{1:T}K_{X_{0}}$. Indeed, for any $\alpha\in\mathbb{R}^{d}$, $d=\dim \Ker(X_{0})$, by linearity of the system, we can rewrite the output data $Y_{T}$ as a sum of a term depending only on $X_{0}$ (free response) and a term depending only on $U_{0:T-1}$ (forced response), so that
\begin{linenomath*}
\begin{equation}
Y_{1:T} K_{X_{0}}\alpha= (GX_{0} + H U_{0:T-1})K_{X_{0}}\alpha = H U_{0:T-1}K_{X_{0}}\alpha,
\end{equation}
\end{linenomath*}
 where $G$, $H$ are matrices of appropriate dimensions that depend on network matrices $A$, $B$, $C$. Thus, assuming that there exists a vector $\alpha^{\star}$ such that $u^{\star}_{0:T-1}=U_{0:T-1}K_{X_{0}}\alpha^{\star}$ is the optimal control reaching $\subscr{y}{f}$ in $T$ steps, it holds that
\begin{linenomath*}
\begin{equation}\begin{aligned}
	\alpha^{\star} = \arg\min_{\alpha}\ \ &  \left\|LK_{X_{0}}\alpha\right\|^{2}_{2} \\
	&\quad  \text{s.t. }  \subscr{y}{f}=Y_TK_{X_{0}} \alpha.
\end{aligned}
\end{equation}
\end{linenomath*}
Hence, along the same lines as before, the optimal control input $u^{\star}_{0:T-1}$ to reach $\subscr{y}{f}$ in $T$ steps is given by 
\begin{linenomath*}
\begin{equation} \begin{aligned}\label{eq:dd-x0}
	u_{0:T-1}^{\star} = U_{0:T-1}K_{X_{0}}(I-K_{Y_{T}K_{X_{0}}}(LK_{Y_{T}K_{X_{0}}})^{\dag} L)(Y_{T}K_{X_{0}})^{\dag}\subscr{y}{f},
\end{aligned}
\end{equation}
\end{linenomath*}
where $K_{Y_{T}K_{X_{0}}}$ is a matrix whose columns form a basis of $\Ker(Y_{T}K_{X_{0}})$. Notice that, if $U_{0:T-1}K_{X_{0}}$ is full row rank, then there always exists a vector $\alpha^{\star}$ such that $u^{\star}_{0:T-1}=U_{0:T-1}K_{X_{0}}\alpha^{\star}$, for any target $\subscr{y}{f}$. A sufficient (yet not necessary) condition for $U_{0:T-1}K_{X_{0}}$ to be full row rank is that \smash{$[U_{0:T-1}^{\transpose} \ X_{0}^{\transpose}]^{\transpose}$} is full row rank.\footnote{Indeed, if \smash{$[U_{0:T-1}^{\transpose} \ X_{0}^{\transpose}]^{\transpose}$} is full row rank, for all $u\in\real^{mT}$ there exists $\gamma\in\Ker(X_{0})$ such that \smash{$[u^{\transpose}\ 0^{\transpose}]^{\transpose} =[U_{0:T-1}^{\transpose} \ X_{0}^{\transpose}]^{\transpose}\gamma$}, which implies that $U_{0:T-1}K_{X_{0}}$ must be of full row~rank.} This in turn implies that (at most) $mT+n$ linearly independent experiments (w.r.t.~to inputs and initial states) suffice to reconstruct the optimal control~via~\eqref{eq:dd-x0}.


\section{Closed-form expression for $Q=0$ and $R=I$ (minimum-energy control input)} For $Q=0$ and $R=I$, we have that $L=I$ and the data-driven expression in \eqref{eq:dd} becomes
\begin{linenomath*}
\begin{equation} \begin{aligned}\label{eq:dd-min-energy}
	u_{0:T-1}^{\star}  = (I-U_{0:T-1}K_{Y_{T}}(U_{0:T-1}K_{Y_{T}})^{\dag})U_{0:T-1}Y_{T}^{\dag}\subscr{y}{f}.
\end{aligned}
\end{equation}
\end{linenomath*}
\eqref{eq:dd-min-energy} equals the minimum-energy control input to reach $\subscr{y}{f}$ in $T$ steps \cite{TK:80}, if $mT$ linearly independent input experiments are collected. \eqref{eq:dd-min-energy} can be further simplified by exploiting the following instrumental result.
\begin{lemma}\label{lemma:pinv-min-energy}
Let $A\in\mathbb{R}^{r \times n}$, $B\in\mathbb{R}^{q \times n}$. If $\text{Ker}(A) \subseteq \text{Ker}(B)$, then $(I-AK_{B}(AK_{B})^{\dag})AB^{\dag}=(BA^{\dag})^{\dag}$, where $K_{B}$ is a matrix whose columns are a basis of $\Ker(B)$.
\end{lemma}
\begin{proof} We show that $(I-AK_{B}(AK_{B})^{\dag})AB^{\dag}$ satisfies the four conditions \cite{AB-TNEG:03} defining the Moore--Penrose pseudoinverse of $BA^{\dag}$. To this aim, let $\Pi := I - AK_{B} (AK_{B})^\dag$. By noticing that $\Pi=\Pi^{\transpose}$ is the orthogonal projection onto $\text{Ker}((AK_{B})^{\transpose})$, we have
\begin{linenomath*}
  \begin{equation}\label{eq:eq_pf_prop_1}
    (AK_{B})^{\transpose} \Pi = 0 \implies \Pi AK_{B} = 0 \implies \Pi AB^{\dag} B = \Pi A.
  \end{equation}
\end{linenomath*}
  By assumption $\text{Ker}(A) \subseteq \text{Ker}(B)$, which implies
\begin{linenomath*}
  \begin{equation}\label{eq:eq_pf_prop_2}
    B (I-A^{\dag} A) = 0 \implies B A^{\dag} A = B,
  \end{equation}
\end{linenomath*}
  since $I-A^{\dag} A$ is the orthogonal projection onto $\text{Ker}(A)$.
  Further, since $BK_{B} = 0$, we have
\begin{linenomath*}
  \begin{equation}\label{eq:eq_pf_prop_3}
    B A^{\dag} (I-\Pi)  = B A^{\dag} AK_{B} (AK_{B})^{\dag} \overset{\eqref{eq:eq_pf_prop_2}}{=} BK_{B} (AK_{B})^{\dag} = 0.
  \end{equation}
\end{linenomath*}
  Finally, since $I-A A^{\dag}$ equals the orthogonal projection onto $\text{Ker}(A^{\transpose})$, and $I-\Pi=AK_{B} (AK_{B})^\dag$ equals the orthogonal projection onto $\text{Im}(AK_{B}) \subseteq \text{Im}(A) \perp \text{Ker}(A^{\transpose})$, where $\text{Im}(\cdot)$ denotes the image or column space of a matrix, we have
\begin{linenomath*}
  \begin{equation} \label{eq:eq_pf_prop_4}
(I-\Pi)(I-A A^{\dag}) = [(I-\Pi)(I-A A^{\dag})]^{\transpose} =0 \implies A A^{\dag} \Pi = \Pi A A^{\dag},
  \end{equation}
\end{linenomath*}
  where the last implication follows because $I-\Pi$ and
  $I-A A^{\dag}$ are symmetric. To conclude, we show that
  $\Pi AB^{\dag} = (BA^{\dag})^\dag$ by proving the four
    Moore--Penrose conditions \cite{AB-TNEG:03}: 
  \begin{enumerate} 
  \item $\Pi AB^{\dag} B A^{\dag} \Pi AB^{\dag} 
    \overset{\eqref{eq:eq_pf_prop_1}}{=}  \Pi A A^{\dag} \Pi AB^{\dag}
    \overset{\eqref{eq:eq_pf_prop_4}}{=} \Pi ^2 A  A^{\dag} A B^{\dag} =
    \Pi AB^{\dag}$; 

  \item
    $BA^{\dag} PAB^{\dag} B A^{\dag}
    \overset{\eqref{eq:eq_pf_prop_1}}{=} B A^{\dag} \Pi A A^{\dag} = B
    A^{\dag} A A^{\dag} - B A^{\dag} (I-\Pi )A A^{\dag}
    \overset{\eqref{eq:eq_pf_prop_3}}{=} B A^{\dag}$;

  \item
    $B A^{\dag} \Pi AB^{\dag} = B A^{\dag} AB^{\dag} - B A^{\dag}(I-\Pi )AB^{\dag} \hspace*{2pt}\overset{\eqref{eq:eq_pf_prop_2},\, 
      \eqref{eq:eq_pf_prop_3}}{=} BB^{\dag} = (BB^{\dag})^{\transpose}$;

  \item
    $\Pi AB^{\dag} B A^{\dag} \overset{\eqref{eq:eq_pf_prop_1}}{=} \Pi A
    A^{\dag} \overset{\eqref{eq:eq_pf_prop_4}}{=} A A^{\dag} \Pi = (\Pi A
    A^{\dag})^{\transpose}$.
  \end{enumerate}
  This concludes the proof.
\end{proof}

Since $Y_{T}=\mc C_{T}U_{0:T-1}$, where $\mc C_{T}=[CB \ CAB \ \cdots\ CA^{T-1}B]$ is the $T$-steps output controllability matrix of the network, it holds that $\text{Ker}(U_{0:T-1}) \subseteq \text{Ker}(Y_{T})$. Thus, by Lemma \ref{lemma:pinv-min-energy}, \eqref{eq:dd-x0} can be compactly rewritten as
\begin{linenomath*}
\begin{equation} \begin{aligned}\label{eq:dd-min-energy-compact}
	u_{0:T-1}^{\star} = (Y_{T}U_{0:T-1}^{\dag})^{\dag}\subscr{y}{f}.
\end{aligned}
\end{equation}
\end{linenomath*}
When the optimal input can be reconstructed from the available data, \eqref{eq:dd-min-energy-compact} could also be derived by ``direct'' estimation of the output controllability matrix $\mc C_{T}$. However, we remark that, based on Lemma \ref{lemma:pinv-min-energy}, the data-driven expressions in \eqref{eq:dd-x0} and \eqref{eq:dd-min-energy-compact} are equivalent even when the optimal input cannot be reconstructed from the available data.

\section{Approximate data-driven minimum-energy controls} Consider the data-driven control input 
\begin{linenomath*}
\begin{equation} \begin{aligned}\label{eq:dd-min-energy-approx}
	\hat u_{0:T-1}  = U_{0:T-1}Y_{T}^{\dag}\subscr{y}{f}.
\end{aligned}
\end{equation}
\end{linenomath*}
Notice that $\hat u_{0:T-1}$ correctly steers the network to $\subscr{y}{f}$ in $T$ steps, as long as $p$ linearly independent experiments are available. Indeed, if $Y_{T}$ is full row rank, there exists $\bar\alpha\in\mathbb{R}^{N}$ satisfying $\subscr{y}{f}=Y_{T}\bar \alpha$, for all $\subscr{y}{f}$, so that $\hat u_{0:T-1}= U_{0:T-1}\bar\alpha=U_{0:T-1}Y_{T}^{\dag}\subscr{y}{f}$ drives the network to $\subscr{y}{f}$. Although $\hat u_{0:T-1}$ does not typically coincide with the minimum-energy control, when the input experiments are generated randomly and independently from a Gaussian distribution, $\hat u_{0:T-1}$ approaches the minimum-energy control as the number of experiments grows, as we show next. To this end we need the following standard result in non-asymptotic random matrix theory, e.g., see \cite[Corollary 5.35 and Lemma 5.36]{RV:12}. Given a matrix $X\in\mathbb{R}^{n\times m}$, $\sigma_{\min}(A)$, $\sigma_{\max}(A)$, and $\kappa(A):=\sigma_{\max}(A)/\sigma_{\min}(A)$ denote the largest, smallest (non-zero) singular value, and the condition number of $A$, respectively.

\begin{lemma}\label{lemma:non-asym-bound}
Let $X\in\mathbb{R}^{N\times q}$ have i.i.d.~normally distributed entries. Then, with probability at least $1-\delta$
\begin{linenomath*}
\begin{eqnarray*}
	\left\|\frac{1}{N}X^{\transpose}X-I\right\|_{2} &\le& 3\max(\eta,\eta^{2}),\\
	1-\eta\le\sigma_{\min}\left(\frac{1}{\sqrt{N}}X\right)&\le&\sigma_{\max}\left(\frac{1}{\sqrt{N}}X\right)\le1+\eta,
\end{eqnarray*}
where $\eta:=\sqrt{q/N}+\sqrt{2\ln(1/\delta)/N}$ and $\|X\|_{2}=\sigma_{\text{max}}(X)$ denotes the spectral norm of matrix $X$.
\end{linenomath*}
\end{lemma}
\begin{theorem}\label{thm:dd-min-energy-approx}
Assume that the network is output controllable, $U_{0:T-1}$ has full row rank, and the entries of $U_{0:T-1}$ are i.i.d.~Gaussian random variables with zero mean and finite variance $\sigma^{2}\ne 0$. Then, with probability at least $1-\delta$
\begin{linenomath*}
\begin{equation}\label{eq:nonasym-approx}
	\|u_{0:T-1}^{\star}-\hat u_{0:T-1}\|_{2} \le \frac{3\max(\eta,\eta^{2})}{\sigma_{\min}(\mc C_{T})}\left(1+\frac{1+\eta}{1-\eta}\kappa^{2}(\mc C_{T})\right)\left\|\subscr{y}{f}\right\|_{2},
\end{equation}
\end{linenomath*}
where $u_{0:T-1}^{\star}$ is the minimum-energy control input driving the network to $\subscr{y}{f}$, the matrix $\mc C_{T}=[CB \ CAB \ \cdots\ CA^{T-1}B]$ is the $T$-steps output controllability matrix of the network, and $\eta:=\sqrt{mT/N}+\sqrt{2\ln(1/\delta)/N}$.
In particular, as $N\to \infty$, 
\begin{linenomath*}
\begin{equation}
	\hat u_{0:T-1} \xrightarrow{\text{a.s.}} u_{0:T-1}^{\star},
\end{equation}
\end{linenomath*}
where $\xrightarrow{\text{a.s.}}$ stands for almost sure convergence.
\end{theorem}
\begin{proof} Since $Y_{T} = \mc C_{T}U_{0:T-1}$, the data-driven input in \eqref{eq:dd-min-energy-approx} can be written as
\begin{linenomath*}
\begin{equation}
	\hat u_{0:T-1} = U_{0:T-1}(\mc C_{T}U_{0:T-1})^{\dag} \subscr{y}{f} = \frac{1}{\sigma^{2}N}U_{0:T-1}U_{0:T-1}^{\transpose} \mc C_{T}^{\transpose}\left(\frac{1}{\sigma^{2}N} \mc C_{T}U_{0:T-1}U_{0:T-1}^{\transpose}\mc C_{T}^{\transpose}\right)^{-1}\subscr{y}{f},
\end{equation}
\end{linenomath*}
where we used that $X^{\dag}=X^{\transpose}(XX^{\transpose})^{-1}$, if $X$ has full row rank, e.g., see \cite{AB-TNEG:03}. By using the above expression and the fact that $u_{0:T-1}^{\star} = \mc C_{T}^{\dag}\subscr{y}{f}$ \cite{TK:80}, we have
\begin{linenomath*}
\begin{equation}
	e:=u_{0:T-1}^{\star}-\hat u_{0:T-1} = \left[\mc C_{T}^{\dag} -\frac{1}{\sigma^{2}N}U_{0:T-1}U_{0:T-1}^{\transpose} \mc C_{T}^{\transpose}\left(\frac{1}{\sigma^{2}N}\mc C_{T}U_{0:T-1}U_{0:T-1}^{\transpose}\mc C_{T}^{\transpose}\right)^{-1}\right]\subscr{y}{f}.
\end{equation}
\end{linenomath*}
By defining the matrix $V_{0:T-1}:=U_{0:T-1}U_{0:T-1}^{\transpose} - \sigma^{2}NI$, the latter equation can be written as
\begin{linenomath*}
\begin{eqnarray}
	e &=& \left[\mc C_{T}^{\dag} -\frac{1}{\sigma^{2}N}(V_{0:T-1}+\sigma^{2}NI) \mc C_{T}^{\transpose}\left(\frac{1}{\sigma^{2}N}\mc C_{T}V_{0:T-1} \mc C_{T}^{\transpose}+\mc C_{T}\mc C_{T}^{\transpose}\right)^{-1}\right]\subscr{y}{f}\notag \\
	&=& \bigg[\mc C_{T}^{\dag} -\frac{1}{\sigma^{2}N}(V_{0:T-1}+\sigma^{2}NI) \mc C_{T}^{\transpose}\left((\mc C_{T}\mc C_{T}^{\transpose})^{-1}-\bigg(\frac{1}{\sigma^{2}N}\mc C_{T}V_{0:T-1} \mc C_{T}^{\transpose}+\mc C_{T}\mc C_{T}^{\transpose}\right)^{-1}\cdot\notag \\
	& \cdot& \frac{1}{\sigma^{2}N}\mc C_{T}V_{0:T-1} \mc C_{T}^{\transpose}(\mc C_{T}\mc C_{T}^{\transpose})^{-1}\bigg)\bigg]\subscr{y}{f} \notag\\
	&=& \left[-I +\frac{1}{\sigma^{2}N}U_{0:T-1}U_{0:T-1}^{\transpose} \mc C_{T}^{\transpose}\left(\frac{1}{\sigma^{2}N}\mc C_{T}U_{0:T-1}U_{0:T-1}^{\transpose} \mc C_{T}^{\transpose}\right)^{-1}\mc C_{T}\right]\frac{1}{\sigma^{2}N}V_{0:T-1} \mc C_{T}^{\dag}\subscr{y}{f}, \label{eq:error}
\end{eqnarray}
\end{linenomath*}
where in the second step we used the matrix identity $(X+Y)^{-1}=Y^{-1}-(X+Y)^{-1}XY^{-1}$, which holds for square matrices $X$, $Y$ with $Y$ and $X+Y$ being non-singular (e.g., see \cite[p.~151]{SSR:17}), and in the last step the identity $\mc C_{T}^{\dag} =\mc C_{T}^{\transpose}(\mc C_{T}\mc C_{T}^{\transpose})^{-1}$ which follows from the fact that $\mc C_{T}$ has full row rank because the network is output controllable by assumption. Thus, from \eqref{eq:error}, the triangle inequality and the submultiplicativity of the~2-norm:
\begin{linenomath*}
\begin{eqnarray}
	\|e\|_{2} &\le& \left(1 +\left\|\frac{1}{\sigma^{2}N}U_{0:T-1}U_{0:T-1}^{\transpose} \mc C_{T}^{\transpose}\left(\frac{1}{\sigma^{2}N}\mc C_{T}U_{0:T-1}U_{0:T-1}^{\transpose} \mc C_{T}^{\transpose}\right)^{-1}\mc C_{T}\right\|_{2}\right)\left\|\frac{1}{\sigma^{2}N}V_{0:T-1} \mc C_{T}^{\dag}\subscr{y}{f}\right\|_{2} \notag \\
	&\le& \left(1 +\left\|\frac{1}{\sigma^{2}N}U_{0:T-1}U_{0:T-1}^{\transpose}\right\|_{2} \left\|\mc C_{T}\right\|_{2}^{2}\left\|\left(\frac{1}{\sigma^{2}N}\mc C_{T}U_{0:T-1}U_{0:T-1}^{\transpose} \mc C_{T}^{\transpose}\right)^{-1}\right\|_{2}\right)\left\|\frac{1}{\sigma^{2}N}V_{0:T-1}\right\|_{2}\cdot\notag \\
	& \cdot&  \left\|\mc C_{T}^{\dag}\right\|_{2} \left\|\subscr{y}{f}\right\|_{2} \notag\\
	&\le& \left(1 + \frac{\sigma_{\max}\left(\frac{1}{\sqrt{\sigma^{2}N}}U_{0:T-1}\right)\sigma^{2}_{\max}(\mc C_{T})}{\sigma_{\min}\left(\frac{1}{\sqrt{\sigma^{2}N}}U_{0:T-1}\right)\sigma^{2}_{\min}(\mc C_{T})}\right)\frac{\left\|\frac{1}{\sigma^{2}N}V_{0:T-1}\right\|_{2}}{\sigma_{\min}(\mc C_{T})} \left\|\subscr{y}{f}\right\|_{2}, \label{eq:error2}
\end{eqnarray}
\end{linenomath*}
where in the last step we used that $\sigma_{\min}(XY)\ge \sigma_{\min}(X)\sigma_{\min}(Y)$\footnote{Indeed, if $X$, $Y$ have full row rank, it holds $\sigma_{\min}(XY) = \lambda_{\min}(XYY^{\transpose}X^{\transpose})\ge \lambda_{\min}(YY^{\transpose})\lambda_{\min}(XX^{\transpose})=\sigma_{\min}(X)\sigma_{\min}(Y)$, where $\lambda_{\min}(\cdot)$ denotes the smallest eigenvalue of a symmetric matrix and we used that $P\succeq \lambda_{\min}(P)I$ if $P\succeq 0$.} and $\|X^{\dag}\|=\sigma_{\min}^{-1}(X)$, for matrices $X$, $Y$ with full row rank. The result now follows from \eqref{eq:error2}, by invoking Lemma \ref{lemma:non-asym-bound}.
\end{proof}
From the non-asymptotic bound in \eqref{eq:nonasym-approx} of Theorem \ref{thm:dd-min-energy-approx}, for a fixed number $N$ of i.i.d.~Gaussian data, the larger $\sigma_{\min}(\mc C_{T})$ is, the closer the data-driven input in \eqref{eq:dd-min-energy-approx} to the minimum-energy one is. Since $\sigma_{\min}^{-1}(\mc C_{T})$ equals the worst-case control energy required to reach a unit-norm target \cite{FP-SZ-FB:13q}, it follows that networks that are ``easy'' to control (i.e., networks featuring a large $\sigma_{\min}(\mc C_{T})$) yield the most favorable approximation performance. In other words, the more ``excitable'' the network dynamics \cite{SD-HM-NM-BR-ST:18} (i.e., the larger $\sigma_{\min}(\mc C_{T})$) are, the lower the approximation error is.

\section{Data-driven optimal control inputs with noisy data}

\subsection{Data corrupted by small noise}

Consider the minimum-energy data-driven expressions \eqref{eq:dd-min-energy-compact}, \eqref{eq:dd-min-energy-approx}, and assume that the data matrices $U_{0:T-1}$, $Y_{T}$ have full (row) rank. Since the Moore--Penrose pseudoinverse of a full (row or column) rank matrix $X$ is a continuous function of the entries of $X$ (in the set of matrices preserving the rank of $X$) \cite[Ch.~6]{AB-TNEG:03}, it follows that \eqref{eq:dd-min-energy-compact} and \eqref{eq:dd-min-energy-approx} are continuous functions of the data matrices around their true values. Thus, small perturbations of the entries of $U_{0:T-1}$, $Y_{T}$, yield a small deviation of the data-driven expressions \eqref{eq:dd-min-energy-compact} and \eqref{eq:dd-min-energy-approx} from their correct values. A similar argument applies to the optimal data-driven control \eqref{eq:dd}, provided that the singular values of the pseudoinverse of $LK_{Y_{T}}$ (which is not typically of full rank) are truncated by small constant $\varepsilon>0$ to preserve the rank of $LK_{Y_{T}}$ when small perturbations are applied to the data matrices $Y_{1:T-1}$ and $Y_{T}$.

\subsection{Data corrupted by i.i.d. noise with zero mean and known variance}
We assume that the data matrices $U_{0:T-1}$, $Y_{1:T-1}$, $Y_{T}$ are corrupted by i.i.d.~noise with zero mean and finite variance. Namely, we consider the following dataset 
\begin{linenomath*} 
\begin{equation} 
\begin{aligned}\label{eq:data-noise}
	U_{0:T-1} &= \bar U_{0:T-1} +\Delta_{U}, \\
	Y_{1:T-1} &= \bar Y_{1:T-1} + \Delta_{Y}, \\ 
	Y_{T} &= \bar Y_{T} + \Delta_{Y_{T}},
\end{aligned}
\end{equation}
\end{linenomath*}
where $\bar U_{0:T-1}$, $\bar Y_{1:T-1}$, and $\bar Y_{T}$ denote the ground truth values, whereas $\Delta_{U}$, $\Delta_{Y}$, and $\Delta_{Y_{T}}$ are independent random matrices with i.i.d.~entries with zero mean and variance $\sigma_{U}^{2}$, $\sigma_{Y}^{2}$, and $\sigma_{Y_{T}}^{2}$, respectively. 

The data-driven controls in \eqref{eq:dd}, \eqref{eq:dd-min-energy}, \eqref{eq:dd-min-energy-compact}, and \eqref{eq:dd-min-energy-approx} computed from the noisy data in \eqref{eq:data-noise} are typically biased and do not converge to the true control input as the data size $N$ grows to infinity. For a concrete example of the latter fact, consider the approximate data-driven control in \eqref{eq:dd-min-energy-approx}, the scalar ($p=m=1$) system $x(t+1)=Ax(t)+u(t)$, $y(t)=x(t)$, and a unitary control horizon ($T=1$). In this simple scenario, \eqref{eq:dd-min-energy-approx} simplifies to
\begin{linenomath*}
\begin{equation} 
 \hat u_{0} = \frac{\sum_{i=1}^{N} u_{1}^{(i)} y_{1}^{(i)}}{\sum_{i=1}^{N} (y_{1}^{(i)})^{2}} = \frac{\sum_{i=1}^{N} (\bar u_{1}^{(i)}+\delta_{U}^{(i)}) (\bar y_{1}^{(i)}+\delta_{Y_{T}}^{(i)})}{\sum_{i=1}^{N} (\bar y_{1}^{(i)}+\delta_{Y_{T}}^{(i)})^{2}},
 \end{equation}
\end{linenomath*} 
where $\bar U_{0}=[\bar u_{0}^{(1)} \cdots \bar u_{0}^{(1)}]$, $\bar Y_{1}=[\bar y_{1}^{(1)} \cdots \bar y_{1}^{(N)}]$, and $\Delta_{U}=[\delta_{U}^{(1)} \cdots \delta_{U}^{(N)}]$, $\Delta_{Y_{T}}=[\delta_{Y_{T}}^{(1)} \cdots \delta_{Y_{T}}^{(N)}]$ denote the true data and noise samples, respectively. By the Strong Law of Large Numbers \cite{AWV00} and the assumption on the noise, as $N\to \infty$, it follows that
\begin{linenomath*}
\begin{equation} 
 \hat u_{0} = \frac{\frac{1}{N}\sum_{i=1}^{N} (\bar u_{1}^{(i)}+\delta_{U}^{(i)}) (\bar y_{1}^{(i)}+\delta_{Y_{T}}^{(i)})}{\frac{1}{N}\sum_{i=1}^{N} (\bar y_{1}^{(i)}+\delta_{Y_{T}}^{(i)})^{2}}\ \  \xrightarrow{\text{a.s.}} \ \  \frac{\frac{1}{N}\sum_{i=1}^{N} \bar u_{1}^{(i)}\bar y_{1}^{(i)}}{\frac{1}{N}\sum_{i=1}^{N} (\bar y_{1}^{(i)})^{2}+\sigma_{Y_{T}}^{2}}.
 \end{equation}
\end{linenomath*} 
Because of the variance term $\sigma_{Y_{T}}^{2}$ in the denominator, $\hat u_{0}$ does not converge to the noiseless control input. To remedy this situation, one could modify the data-driven expressions in order to compensate for the variance of the noise, as we detail next. 

We first consider the data-driven control in \eqref{eq:dd} and rewrite it as
\begin{linenomath*}
\begin{equation} \begin{aligned}\label{eq:dd-rew}
	\hat u_{0:T-1} &= U_{0:T-1}\left(I-K_{Y_{T}}\left(LK_{Y_{T}}\right)^{\dag} L^{\transpose} \right)Y_{T}^{\dag}\subscr{y}{f}  \\
	&= U_{0:T-1}\left(I-\Pi_{Y_{T}}L^{\transpose} \left(L\Pi_{Y_{T}}L^{\transpose} \right)^{\dag} L \right)Y_{T}^{\transpose}(Y_{T}Y_{T}^{\transpose})^{\dag}\subscr{y}{f},
\end{aligned}
\end{equation}
\end{linenomath*}
where $\Pi_{Y_{T}}=K_{Y_{T}}K_{Y_{T}}^{\transpose}=I-Y_{T}^{\dag}Y_{T}=I-Y_{T}^{\transpose}(Y_{T}Y_{T}^{\transpose})^{\dag}Y_{T}$ denote the orthogonal projection onto $\Ker(Y_{T})$ and we used that $X^{\dag}=X^{\transpose}(XX^{\transpose})^{\dag}$, for any matrix $X$, e.g., see \cite{AB-TNEG:03}. Next, we consider the following ``corrected'' version of \eqref{eq:dd-rew}
\begin{linenomath*}
\begin{align} \begin{aligned}\label{eq:dd-corr}
	\hat u^{\text{(c)}}_{0:T-1} &= U_{0:T-1}\left(I-\tilde\Pi_{Y_{T}}L^{\transpose} \left(L \Pi_{Y_{T}}L^{\transpose}-\begin{bmatrix} N\sigma_{Y}^{2}Q & 0 \\ 0 & N\sigma_{U}^{2}R\end{bmatrix}\right)_{\varepsilon}^{\dag} L \right)Y_{T}^{\transpose}(Y_{T}Y_{T}^{\transpose}-N\sigma_{Y_{T}}^{2}I)^{\dag}\subscr{y}{f},
\end{aligned}
\end{align}
\end{linenomath*}
where $\tilde\Pi_{Y_{T}}=I-Y_{T}^{\transpose}(Y_{T}Y_{T}^{\transpose}-N\sigma_{Y_{T}}^{2}I)^{\dag}Y_{T}$, $L$ is the particular square root 
\begin{linenomath*}
\begin{equation} 
L:=\begin{bmatrix} Q^{1/2} Y_{1:T-1} \\ R^{1/2}  U_{0:T-1}\end{bmatrix},
\end{equation}
\end{linenomath*}
and $X^{\dag}_{\varepsilon}$ denotes the Moore--Penrose pseudoinverse of $X$ that treats as zero the singular values of $X$ that are smaller than $\varepsilon>0$.

\begin{theorem}\label{thm:dd-min-energy-approx} Consider the noisy dataset as in \eqref{eq:data-noise} and assume that $\bar U_{0:T-1}$ has full row rank. For $\varepsilon>0$ sufficiently small and $N\to\infty$, the control sequence in \eqref{eq:dd-corr} converges almost surely to the optimal control input; that is,
\begin{linenomath*}
\begin{equation}\label{eq:dd-corr-conv}
	\hat u^{\text{(c)}}_{0:T-1} \xrightarrow{\text{a.s.}} u_{0:T-1}^{\star}.
\end{equation}
\end{linenomath*}
\end{theorem}
\begin{proof} 
After some algebraic manipulations, \eqref{eq:dd-corr} can be written as
\begin{linenomath*}
\begin{eqnarray}\label{eq:dd-corr-exp}
\hat u^{\text{(c)}}_{0:T-1}  &=&  U_{0:T-1}\left(I-\tilde\Pi_{Y_{T}}L^{\transpose} \left(L \Pi_{Y_{T}}L^{\transpose}-\begin{bmatrix} N\sigma_{Y}^{2}Q & 0 \\ 0 & N\sigma_{U}^{2}R\end{bmatrix}\right)_{\varepsilon}^{\dag} L \right)Y_{T}^{\transpose}(Y_{T}Y_{T}^{\transpose}-N\sigma_{Y_{T}}^{2}I)^{\dag}\subscr{y}{f}\notag \\
& =&U_{0:T-1}Y_{T}^{\transpose}(Y_{T}Y_{T}^{\transpose}-N\sigma_{Y_{T}}^{2}I)^{\dag}\subscr{y}{f} - U_{0:T-1}\tilde\Pi_{Y_{T}}L^{\transpose}  \cdot \notag\\
&\cdot& \left(L \Pi_{Y_{T}}L^{\transpose}-\begin{bmatrix} N\sigma_{Y}^{2}Q & 0 \\ 0 & N\sigma_{U}^{2}R\end{bmatrix}\right)_{\varepsilon}^{\dag} L Y_{T}^{\transpose}(Y_{T}Y_{T}^{\transpose}-N\sigma_{Y_{T}}^{2}I)^{\dag}\subscr{y}{f}\notag \\
& =&P_{2}(P_{1}-\sigma_{Y_{T}}^{2}I)^{\dag}\subscr{y}{f} - \left(P_{3} - P_{2}(P_{1}-\sigma_{Y_{T}}^{2}I)^{\dag}P_{4} \right)\cdot \notag\\
&\cdot&\left(P_{4}^{\transpose}(P_{1}-\sigma_{Y_{T}}^{2}I)^{\dag}P_{4}+ P_{5}-\begin{bmatrix} \sigma_{Y}^{2}Q & 0 \\ 0 & \sigma_{U}^{2}R\end{bmatrix}\right)_{\varepsilon}^{\dag} P_{4}^{\transpose}(P_{1}-\sigma_{Y_{T}}^{2}I)^{\dag}\subscr{y}{f}
\end{eqnarray}
\end{linenomath*}
where $P_{1} = \frac{1}{N} Y_{T}Y_{T}^{\transpose}$, $P_{2} = \frac{1}{N}U_{0:T-1}Y_{T}^{\transpose}$, $P_{3} = \frac{1}{N}U_{0:T-1}L^{\transpose}$, $P_{4} = \frac{1}{N}Y_{T}L^{\transpose}$, and $P_{5} = \frac{1}{N}L L^{\transpose}$.  By the Strong Law of Large Numbers \cite{AWV00} and the assumption on the noise, as $N\to \infty$, it follows that
\begin{linenomath*}
\begin{eqnarray}\label{eq:LLN}
\begin{split}
	&P_{1} = \frac{1}{N} Y_{T}Y_{T}^{\transpose} \xrightarrow{\text{a.s.}}  \frac{1}{N}\bar{Y}_{T}\bar{Y}_{T}^{\transpose} + \sigma_{Y_{T}}^{2}I =:\bar P_{1}, \\&
	P_{2} = \frac{1}{N}U_{0:T-1}Y_{T}^{\transpose} \xrightarrow{\text{a.s.}}  \frac{1}{N}\bar U_{0:T-1}\bar Y_{T}^{\transpose}=:\bar P_{2}, \\
	&P_{3} = \frac{1}{N}U_{0:T-1}L^{\transpose} \xrightarrow{\text{a.s.}} \frac{1}{N} \bar U_{0:T-1}\begin{bmatrix}  \bar Y_{1:T-1}^{\transpose}Q^{1/2} &  \bar U_{0:T-1}^{\transpose}R^{1/2}\end{bmatrix}=:\bar P_{3}, \\
	&P_{4} = \frac{1}{N}Y_{T}L^{\transpose} \xrightarrow{\text{a.s.}}  \frac{1}{N}\bar Y_{T}\begin{bmatrix}  \bar Y_{1:T-1}^{\transpose}Q^{1/2} &  \bar U_{0:T-1}^{\transpose}R^{1/2}\end{bmatrix}=: \bar P_{4}, \\
	&P_{5} = \frac{1}{N}LL^{\transpose} \xrightarrow{\text{a.s.}} \begin{bmatrix}  \frac{1}{N} Q^{1/2} \bar Y_{1:T-1}  \bar Y_{1:T-1}^{\transpose} Q^{1/2}+ \sigma_{Y}^{2}Q & \frac{1}{N} Q^{1/2} \bar Y_{1:T-1}  \bar U_{0:T-1}^{\transpose} R^{1/2} \\ \frac{1}{N}R^{1/2}\bar U_{0:T-1}  \bar Y_{1:T-1}^{\transpose}Q^{1/2} &  \frac{1}{N}R^{1/2}\bar U_{0:T-1}  \bar U_{0:T-1}^{\transpose}R^{1/2} + \sigma_{U}^{2}R\end{bmatrix}=: \bar P_{5}.
\end{split}
\end{eqnarray}
\end{linenomath*}
Notice that $\hat u^{\text{(c)}}_{0:T-1}$ is a continuous function of $P_{i}$ around $P_{i}=\bar P_{i}$ for $i=1,\dots,5$ and $\varepsilon>0$ sufficiently small. In light of this fact, \eqref{eq:dd-corr-conv} follows by using \eqref{eq:LLN} and the Continuous Mapping Theorem \cite[Theorem 2.3]{AWV00}. 
\end{proof}

Following the same argument as above, it is possible to establish asymptotically correct data-driven expressions of minimum-energy controls ($Q=0$, $R=I$). Specifically, the corrected version of \eqref{eq:dd-min-energy} reads as
\begin{linenomath*}
\begin{equation} \begin{aligned}\label{eq:dd-min-energy-corr}
	\hat u^{\text{(c)}}_{0:T-1}  = (I-U_{0:T-1}\tilde \Pi_{Y_{T}}(U_{0:T-1}\tilde \Pi_{Y_{T}}U_{0:T-1}^{\transpose}-N\sigma_{U}^{2}I)_{\varepsilon}^{\dag})U_{0:T-1}(Y_{T}Y_{T}^{\transpose}-N\sigma_{Y_{T}}^{2}I)^{\dag}\subscr{y}{f},
\end{aligned}
\end{equation}
\end{linenomath*}
whereas the corrected version of the compact data-driven control in \eqref{eq:dd-min-energy-compact} is
\begin{linenomath*}
\begin{equation} \begin{aligned}\label{eq:dd-min-energy-compact-corrected}
	\hat u^{\text{(c)}}_{0:T-1} = (Y_{T}U_{0:T-1}^{\transpose}(U_{0:T-1}U_{0:T-1}^{\transpose}- N\sigma_{U}^{2}I)^{\dag})^{\dag}\subscr{y}{f}.
\end{aligned}
\end{equation}
\end{linenomath*}
Finally, the corrected approximate minimum-energy control in \eqref{eq:dd-min-energy-approx} reads as
\begin{linenomath*}
\begin{equation} \begin{aligned}\label{eq:dd-min-energy-approx-corr}
	\hat u^{\text{(c)}}_{0:T-1}  = U_{0:T-1}(Y_{T}Y_{T}^{\transpose}-N\sigma_{Y_{T}}^{2}I)^{\dag}\subscr{y}{f}.
\end{aligned}
\end{equation}
\end{linenomath*}
It is worth noting that, while \eqref{eq:dd-min-energy-corr} requires correction terms for both input and output noises, \eqref{eq:dd-min-energy-compact-corrected} and \eqref{eq:dd-min-energy-approx-corr} include correction terms only for one source of noise (input noise in \eqref{eq:dd-min-energy-compact-corrected} and output noise in \eqref{eq:dd-min-energy-approx-corr}). In particular, if the noise corrupts output data only, \eqref{eq:dd-min-energy-compact-corrected} coincides with the original data-driven control in \eqref{eq:dd-min-energy-compact}.

\newpage
\bibliography{alias,FP,Main,New}